%% file: TypicalCellSpasWin_Arxiv.tex
\documentclass[conference]{IEEEtran}

\usepackage{tikz}
\usepackage{algpseudocode, algorithm,algorithmicx}
\usepackage{amsmath,bm,bbm,amsthm, amssymb, authblk}
\usepackage{color}
\usepackage{animate}
\usepackage{etoolbox}
\usepackage{ulem}
\usepackage{comment}
\usepackage{pgf}
\usepackage{caption}
\usepackage{subcaption}
\usetikzlibrary{calc}
\usetikzlibrary{patterns}
\usetikzlibrary{arrows}
\usetikzlibrary{decorations.pathreplacing}
\usepackage{fancyvrb}
\usepackage[utf8]{inputenc}
\usepackage{pgfplots}

\newcommand{\wh}{\widehat}

\newcommand{\vXvi}[1]{X_{\ms{V}, #1}}

\newcommand{\C}{\mathcal{C}}

\newcommand{\E}{\mathbb{E}} 
\newcommand{\Z}{\mathbb{Z}}
\newcommand{\R}{\mathbb{R}}

\def\eq{\begin{equation}}
\def\en{\end{equation}}

\newtheorem{theorem}{Theorem}[section]

\newtheorem{proposition}[theorem]{Proposition}

    \def\a{\alpha}

    \def\phi{\varphi}
    \def\g{\gamma}
    \def\la{\lambda}

    \def\s{\sigma}

    \def\x{\xi}

    \def\T{\T}

    \def\V|{{\Vert}}

    \def\d{{\rm d}}
    \def\E{\mathbb{E}}
    \def\V{\mathbb{V}}

    \def\mc{\mathcal}
    
    \def\ms{\mathsf}

    \def\R{\mathbb{R}}
    \def\Z{\mathbb{Z}}

\def\vgv{\g^{\ms{v}}}
\def\vgh{\g^{\ms{h}}}

\def\vI{I}

    \def\vXh{X^{\ms{h}}}
    \def\vXvi{X_i^{\ms{v}}}    
     
    \def\vXv{X^{\ms{v}}}   
    \def\vXhs{X^{\ms{h,*}}}
    \def\vXvs{X^{\ms{v,*}}}

    \def\vSv{S^{\ms{v}}}
    \def\vSh{S^{\ms{h}}}
        \def\vSvs{S^{\ms{v,*}}}
    \def\vShs{S^{\ms{h,*}}}
    \def\vSp{S^{\bullet}}
   	\def\vSvp{S^{\ms{v},\bullet}}
    \def\vShp{S^{\ms{h},\bullet}}
    \def\vgv{\gamma^{\ms{v}}}
    \def\vgh{\gamma^{\ms{h}}}

        \def\vAs{A^*}

\title{The typical cell in anisotropic tessellations}

\author[1]{A.~Hinsen~\thanks{hinsen@wias-berlin.de}}
\author[2]{C.~Hirsch~\thanks{christian@math.aau.dk}}
\author[1]{B.~Jahnel~\thanks{jahnel@wias-berlin.de}}
\author[3]{E.~Cali~\thanks{elie.cali@orange.com}}

\affil[1]{Weierstrass Institute for Applied Analysis and Stochastics, Mohrenstra{\ss}e 39
10117 Berlin, Germany}
\affil[2]{Department of Mathematical Sciences, Aalborg University, Skjernvej 4, 9220 Aalborg, Denmark}
\affil[3]{Orange Labs Network, 44 Avenue de la R\'epublique, 92320 Ch\^atillon, France}

\begin{document}

\maketitle

\begin{abstract}
	The typical cell is a key concept for stochastic-geometry based modeling in communication networks, as it provides a rigorous framework for describing properties of a serving zone associated with a component selected at random in a large network. We consider a setting where network components are located on a large street network. While earlier investigations were restricted to street systems without preferred directions, in this paper we derive the distribution of the typical cell in Manhattan-type systems characterized by a pattern of horizontal and vertical streets. We explain how the mathematical description can be turned into a simulation algorithm and provide numerical results uncovering novel effects when compared to classical isotropic networks.
\end{abstract}

\begin{IEEEkeywords}
Palm calculus, typical cell, Cox point process, random tessellation, anisotropy.
\end{IEEEkeywords}

\section{Introduction}
%
%
In large cities, major network operators manage several hundreds of base stations together with their associated \emph{serving zones}. The first basic desire is then to understand the average behavior of these serving zones. 
In other words, to give answers to the question:
\begin{center}
\emph{What are the characteristics of a typical cell?}
\end{center}

%
%
As an alternative to a slow and error-prone analysis of each individual cell, stochastic geometry offers a versatile modeling framework, which can serve as a predictive tool or to study use cases~\cite{baccelli2009stochastic1, GlEl17}. In this setting, serving zones are represented by the cells of the Voronoi tessellation associated to a point process of base stations. Leveraging the formalism of \emph{Palm calculus} makes it possible to assign a precise meaning to the idea of selecting a cell at random in a possibly unbounded network. 

%
%
A specific difficulty in the setting of telecommunication networks stems from the practice that cell centers are often located along the streets. In the past, this has been particularly relevant in fixed-access networks, where precise mathematical descriptions together with simulation algorithm have been developed for the typical cell in street systems based on classical tessellation models such as the Poisson line, the Poisson-Voronoi or the Poisson-Delaunay tessellation \cite{line, voronoi, delaunay}.  Nowadays, due to the rise of device-to-device (D2D) networks, for example vehicular networks, and in the context of fifth-generation wireless networks (5G), models with nodes located on street systems become increasingly important, see~\cite{choi1, choi2, quentin, d2dProceeding}. Indeed, vehicles are mostly confined to streets and their communication, within the additional 5G frequencies, should be considered outdoors, since these frequencies will hardly be able to penetrate walls. Besides, also the connectivity characteristics of such ad-hoc network models sensitively depend on the specifics of the spatial patterning of nodes, see~\cite{glauche}.

%
%
In view of the increasingly important emerging markets in the Middle East and African, the classical tessellation models are only of limited use, since they share isotropy as a common feature. In other words, there is no single predominant direction along which the roads would align. This contrasts starkly the prevailing street topography in emerging-market economies, which often exhibits a strikingly rectangular or Manhattan-like architecture.

%
%
At the time of writing, only few anisotropic models have been considered as models for telecommunication networks. Poisson line tessellations exhibit an appealingly tractable Palm distribution \cite{line} and can be made anisotropic by a suitable choice of the directional distribution. However, the distances between streets are necessarily exponentially distributed and thereby often exhibit higher variability than what would be expected in practice. On the other side of the spectrum, also entirely rigid lattice models as considered in \cite{glauche} are not ideal, since in reality distances between streets are not fixed in an entire city.

\begin{figure}[!htpb]
\centering
\begin{subfigure}{.25\textwidth}
  \centering
\input{ManhattanPic_1.tex}
\end{subfigure}%
\begin{subfigure}{.25\textwidth}
  \centering
\input{NestedManhattanPic.tex}
\end{subfigure}
 \caption{The typical cell in a Manhattan grid (left) and in a nested Manhattan grid (right)}
\label{palmFig}
\end{figure}
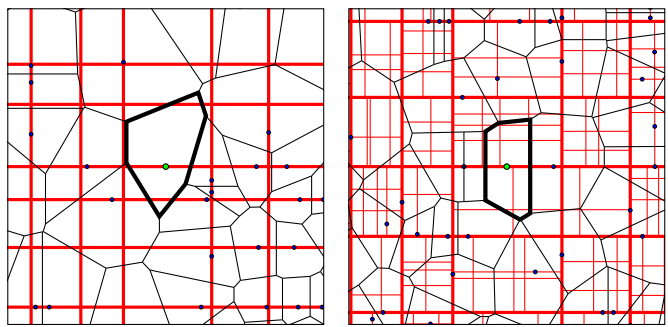

%
%
In light of this discussion, we extend the method of Palm calculus and typical cells to a more flexible class of anisotropic street system. In particular, based on notions from renewal theory, we introduce the \emph{Manhattan grids}, which represent street systems with two perpendicular directions and a flexible distribution for distances between streets. As an additional step, our model extends to a nested setting, where inside blocks of the Manhattan grid, further systems of side streets are added. Figure~\ref{palmFig} illustrates the system model based on Manhattan and nested Manhattan grids.
The main mathematical contributions of this work, Theorems~\ref{MG_Palm_Rep} and~\ref{NMG_Palm_Rep}, provide a simulation algorithm for the typical cell in Manhattan grids based on a tractable description of the Palm distribution of the street system. Based on this algorithm, we show in a simulation study that anisotropy and regularity of Manhattan grids lead to surprising novel effects that are not present in the classical tessellation models.

%
%
The rest of the paper is organized as follows. In Section \ref{modSec}, we provide the precise mathematical details of the proposed network model. In Section \ref{resSec}, we derive a tractable representation of the Palm distribution and a simulation algorithm for the typical cell in Manhattan and nested Manhattan grids. Then, in Section \ref{simSec}, we illustrate in a simulation study that anisotropy can lead to surprising consequences for the distribution of central network characteristics, such as shortest-path lengths. Finally, Section \ref{concSec} summarizes the findings and points to avenues for further research.


\section{System Model}
\label{modSec}

\subsection{Street-system based network models}
To define the typical cell, we consider a network model with components on a street system as described in \cite{line}. More precisely, we start from a random \emph{street system} $S$ that is invariant with respect to translations in the two space dimensions but not necessarily isotropic. Writing $|B|$ for the total street length in an area $B$, we let $\gamma = \E[|S \cap [0, 1]^2|]$ denote the \emph{street intensity} of $S$, i.e., the expected street length per unit area.
Given a realization of the street system, network components $Y = \{Y_i\}_{i \ge 1}$ are scattered at random along the streets according to a Poisson point process with a linear intensity $\lambda > 0$. Since the street system itself comes from a probabilistic model, the network components $Y$ form a Cox point process with random intensity measure concentrated on $S$.

The network component $Y_i$ is responsible for providing service in an area described by its \emph{cell}
$$\C_i =\{y \in \R^2:\, |y - Y_i| \le \inf_{j \ge 1}|y - Y_j|\}.$$ 
The family of these cells forms the so-called \emph{Voronoi tessellation} of $Y$. 

Then, the \emph{typical cell} $\C^*$ encodes the idea of the cell associated with a component selected at random from $Y$. 
More precisely, relying on the machinery of Palm calculus, $\C^*$ is determined via the distributional identity
\begin{align}
	\label{typCellEq}
	\E[f(\C^*)] = \frac1{\lambda \gamma} \E\Big[\sum_{Y_i \in [0, 1]^2}f(\C_i - Y_i)\Big].
\end{align}
for any non-negative measurable test-function $f$.

%
%
\subsection{Definitions of Manhattan- and nested Manhattan grids}
\label{defSec}
In this paper, the underlying street system $S$ forms either a {Manhattan grid} or a {nested Manhattan grid}. First, we describe the construction of the Manhattan grid. Consider two independent, identically distributed and stationary renewal processes $\vXv$ and $\vXh$ on $\R$ representing the coordinates of vertical and horizontal streets, respectively. That is, the distances between the $i$th and the $(i+1)$th point in such a process are independent over $i\in\mathbb Z$ and distributed according to an inter-arrival distribution. It is further assumed to have the property that any translation of the process does not change its distribution. Then, we define the \emph{Manhattan grid} as
$$
S = \vSh \cup \vSv = (\R \times \vXh) \cup (\vXv \times \R).
$$
In particular, the street intensity decomposes as
$$\g = \vgh + \vgv  = \E\big[|\vSh \cap [0, 1]^2|\big] + \E\big[|\vSv \cap [0, 1]^2|\big].$$
In the special case where the inter-arrival distribution is exponential, both $\vXv$ and $\vXh$ are homogeneous Poisson processes, so that we recover a rectangular Poisson line tessellation.  

%
%
\medskip
The construction of the \emph{nested Manhattan grid} $\bar S$, builds on a Manhattan grid $S$ as the initial tessellation of main streets. Then, a sequence of further independent and identically distributed (iid) Manhattan grids $S_1,S_2,\dots$, models the side streets within every cell or block $(\Xi_i)_{i\ge1}$ in $S$. 
More precisely, $\Xi_i \cap S_i$ describes the inner structure of each $\Xi_i$. The street intensity $\bar \g$ of $\bar S$ takes into account contributions both from the main streets as well as the side streets. Therefore, it equals $\bar\g = \g + \g_1$ where 
$$\g_1 =  \E\big[|S_1 \cap [0,1]^2|\big].$$
denotes the intensity of side streets.

\medskip
%
%
\subsection{Construction of the underlying renewal processes}
As explained above, in order to construct and implement a Manhattan grid with given inter-arrival distribution, we first need to be able to simulate the corresponding stationary renewal process. In queuing theory, a \emph{one-sided renewal process} $A = \{A_i\}_{i\ge1}$ with inter-arrival distribution $\mc L(I)$ is a point process on the positive half-line $[0,\infty)$ where the inter-arrival times $\vI_{i+1} = A_{i+1} - A_i$ are iid with distribution $\mc L(I)$, see \cite[Section V.1]{As03}.

In order to transform the one-sided process into a stationary renewal process, we first add another independent renewal process on the negative half-line $(-\infty,0]$ with the same inter-arrival distribution. Although this defines a two-sided process, it is not yet stationary as the initial segment containing the origin requires special attention.  More precisely, due to the waiting-time paradox, when picking a random point in time, the probability for this point to be in a larger segment is increased \cite[Section V.3]{As03}. Therefore, the length of the segment containing the origin follows a \emph{length-biased distribution} $\mc L(\vI^*)$ characterized by
$$\E[f(\vI^*)] = \frac1{\E[\vI]}\E[\vI f(\vI)]$$
for any non-negative measurable test-function $f$. If the distribution $\mc L(\vI)$ has a density $g(x)$ w.r.t.~the Lebesgue measure, then $\mc L(\vI^*)$ has the density $xg(x)/\a$, where $\a$ is the normalization constant.

%
%
\medskip
As illustrated in Figure~\ref{statFig}, we construct the stationary renewal process by sampling the length of the initial segment containing the origin according to the length-biased distribution with the origin chosen uniformly at random on that edge. The remaining segments are sampled independently according to the given inter-arrival distribution $\mc L(\vI)$.  Then, \cite[Theorem 8.4.1]{Th00} describes the distribution of the resulting stationary process, which we present in Proposition~\ref{StatRen}.
\begin{figure}[!htpb]
    \centering
    \input{statFig}
    \caption{Construction of the stationary renewal process}
\label{statFig}
\end{figure}
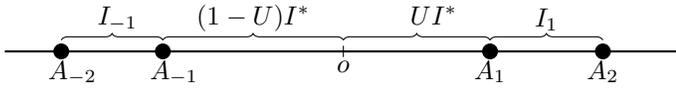

\begin{proposition}\label{StatRen}
	Let $\{\vI_i\}_{i \in \Z \setminus \{0\}}$ be a doubly-infinite sequence of iid inter-arrival times distributed according to $\mc L(I)$. Next, let $\vI^*$ be an independent random variable distributed according to the length-biased distribution and let $U \sim {\bf Unif}([0,1])$ be independent and uniformly distributed on $[0,1]$. Then, the point process of arrival times
    $$A_i=U \vI^*+\sum_{j \le i-1} \vI_j\quad\text{ and }\quad A_{-i}=(U-1)\vI^*-\sum_{j \le i-1}\vI_{-j}$$
    is stationary. It defines the \emph{stationary renewal process with inter-arrival distribution $\mc L(\vI)$}.
\end{proposition}

%
%
When building the typical cell, the network is seen from a network component selected at random on the street system. Since in the Manhattan-grid case, this network is constructed from a stationary renewal process, we need to describe the stationary renewal process seen from a randomly selected arrival time. To make this mathematically precise, we resort to the framework of Palm calculus. More precisely, the \emph{Palm version} $A^*$ of the stationary renewal process $A$ is determined by the distributional identity
$$\E[f(A^*)] = \frac1{\E[\#\{A\, \cap [0,1]\}]}\E\Big[\sum_{A_i \in A\, \cap [0, 1]} f(A - A_i)\Big],$$
for any non-negative measurable test-function $f$. 
In the setting of renewal processes, passing from $A$ to $A^*$ is dual to the construction of $A$ as made precise in \cite[Theorem 8.4.1]{Th00}, which we summarize in Proposition~\ref{PalmRen}.
\begin{proposition}\label{PalmRen}
    Let $A$ be a stationary renewal process with inter-arrival distribution $\mc L(\vI)$. Then, the Palm version $\vAs= \{\vAs_i\}_{i \in \Z}$ is given by $\vAs_0=o$ and for $i \ge 1$,
	\begin{align}
		\label{palmRenEq}
		\vAs_i=\sum_{j \le i}\vI_j\quad \text{and}\quad\vAs_{-i} = -\sum_{j \le i}\vI_{-j}.
	\end{align}
\end{proposition}


\section{Results}
\label{resSec}

In this section, we derive a representation for the typical cell $\C^*$, which is both easy to understand and amenable to simulations. It is based on a tractable description of the Palm versions of the underlying street systems.

\subsection{Distribution of the typical cell}
We rely on \cite[Lemma 3.3]{line}, which expresses the distribution of $\C^*$ as defined in \eqref{typCellEq} with respect to the distribution of the Palm version $S^*$ of the underlying street system. The latter describes the street system seen from a point selected at random on the streets and is formally given as 
\begin{align}
	\label{palmStreetEq}
	\E[f(S^*)] = \frac1\g \E\Big[\int_{S\cap [0, 1]^2}f(S - y) \d y\Big],
\end{align}
for any non-negative measurable test-function $f$. We summarize the result of \cite[Lemma 3.3]{line} in the following simulation algorithm for the typical cell $\C^*$. We first realize a street system $S^*$ distributed according to the Palm version of $S$ and then place the network components according to a Poisson point process $Y$ on this network. The typical cell corresponds to the cell when adding an additional point at the origin to this network.
\begin{algorithm}
  \caption{Typical cell $\C^*$
    \label{cellSim}}
  \begin{algorithmic}
          \State          $S^* \sim$ Palm version of $S$
	  \State          Given $S^*$: $Y\sim$ Poisson point process with intensity $\la |S^* \cap \cdot|$
	  \State \Return Voronoi cell $\C^*$ at $o$ with respect to $\{o\} \cup Y$
	   \end{algorithmic}
\end{algorithm}

Hence, the main mathematical contributions of the present paper are rigorous derivations of tractable representations for the Palm distribution of the Manhattan grid and the nested Manhattan grid.

\subsection{Palm distribution of Manhattan grids}

First, we provide an algorithmic description of a street system $\vSp$ that will be shown to have the same distribution as $S^*$. In words, in the algorithm, the origin is located on a horizontal street with probability $\vgh/\g$. In this case, the coordinates of vertical streets form a stationary renewal process $\vXv$, whereas for the horizontal streets we need the Palm version $\vXhs$ as defined in \eqref{palmRenEq}. If the origin is located on a vertical street, the roles are reversed.

%
%
\begin{algorithm}
  \caption{Palm version of a Manhattan grid
    \label{mgSim}}
  \begin{algorithmic}
          \State          $U \sim$ {\bf Unif}$([0,1])$
        \If{$U \le \vgh/\g$}
         \State $\vSvp\gets\vXv\times\R$ and $\vShp\gets\R\times \vXhs$
          \Else 
         \State $\vSvp\gets\vXvs\times\R$ and $\vShp\gets\R\times \vXh$
        \EndIf
          \State \Return{$\vSvp \cup \vShp$}
  \end{algorithmic}
\end{algorithm}

%
%
\begin{theorem}[Palm version of a Manhattan grid]
	\label{MG_Palm_Rep}
	The Palm distribution of the Manhattan grid can be represented via Algorithm \ref{mgSim}.
\end{theorem}

\begin{proof}
	Starting from definition \eqref{palmStreetEq}, we compute
\begin{align*}
        &\E\big[f(\vSvs,\vShs)\big]\\
	& = \frac 1\g \E\Big[\int_{S\cap[0, 1]^2}f(\vSv - x,\vSh - x)\d x\Big]\\
&= \frac 1\g \E\Big[\int_{\vSv\cap[0, 1]^2}f(\vSv - x,\vSh - x)\d x\Big]\\
        &\phantom{=} + \frac 1\g \E\Big[\int_{\vSh\cap[0, 1]^2}f(\vSv - x,\vSh - x)\d x\Big].
\end{align*}
Then, for the first summand, we further write
\begin{align*}
	&\E\Big[\int_{\vSv\cap[0,1]^2}f(\vSv - x,\vSh - x)\d x\Big]\\
        &=\E\Big[\sum_{\vXvi \in [0,1]} \int_0^1f(\vSv - (\vXvi,u),\vSh - (\vXvi,u))\d u\Big]\\
&= \E\Big[\sum_{\vXvi \in [0,1]} \int_0^1f(\vSv  -  (\vXvi,0),\vSh  -  (0,u))\d u\Big]\\
&= \E\Big[\sum_{\vXvi \in [0,1]} f(\vSv - (\vXvi,0),\vSh)\Big]
\end{align*}
	where we used stationarity in the last step. Of course the computation can be repeated with reversed roles of horizontal and vertical streets.	Hence, by definition \eqref{palmRenEq},
\begin{align*}
        \E\big[f(\vSvs,\vShs)]&=\frac{\vgv}\g \E[f(\vXvs\times\R,\vSh)\big] \\
	& +\frac{\vgh}\g \E[f(\vSv,\R\times\vXhs)]=\E[f(\vSvp,\vShp)]
\end{align*}
which is the desired representation.
\end{proof}

\subsection{Palm distribution of nested Manhattan grids}
Next, we give an algorithm for nested Manhattan grids. Here, we assume that the Palm version of the non-nested Manhattan grid is available via Algorithm \ref{mgSim}.

%
%
\begin{algorithm}
  \caption{Palm version of a nested Manhattan grid
    \label{nmgSim}}
  \begin{algorithmic}
          \State $U \sim$ {\bf Unif}$([0,1])$
        \If{$U \le \g/\bar\g$}
         \State $\vSp \gets S^*$ and  $\vSp_1\gets S_1$
          \Else
         \State $\vSp \gets S$ and $\vSp_1\gets S^*_1$
        \EndIf
          \State \Return{$\vSp \cup \vSp_1$}
  \end{algorithmic}
\end{algorithm}

%
%
\begin{theorem}[Palm version of a nested Manhattan grid]
	\label{NMG_Palm_Rep}
	The Palm distribution of the nested Manhattan grid can be represented via Algorithm \ref{nmgSim}.
\end{theorem}
\begin{proof}
Proceeding along the lines of the proof of Theorem \ref{MG_Palm_Rep}, the Palm version $\bar S^*$ can be written as a
	mixture of two Palm versions. Formally, $\bar S = h(S; \{S_i\}_{i \ge 1})$, where $h$ encodes the translation-covariant construction rule described in Section \ref{defSec}. Then, by~\cite[Theorem 3.1]{nested}, the identity
	\begin{align*}
		\E[f( S^*)] &= \frac{\g}{\bar\g}\E[f(h(S^*;\{S_i\}_{i\ge1})] \\
		&\phantom{=} + \frac{\g_1}{\bar\g}\E[f(h(S;S^*_1,\{S_i\}_{i \ge 2})]
	\end{align*}
holds for all non-negative measurable test-functions $f$, where $S_1$ denotes the side-street model within the block of the main street model containing the origin.
But this is the desired representation.
\end{proof}

%
%
%
%


\section{Simulation study}
\label{simSec}
Algorithms~\ref{cellSim} and~\ref{mgSim} open the door to investigating how the anisotropy of the street system affects pivotal network characteristics. As a prototypical example, we present here \emph{typical shortest-path lengths}. Loosely speaking, this network characteristic describes the length of the shortest path on the street network from a randomly selected point on the street to the network component in whose serving zone it is located. 

More rigorously, we place network components $Z$ according to a Cox process on the Palm distribution $S^*$ of the street system, and then measure the shortest connection length $\ell\big(Z_i\big)$ on $S^*$ from the origin to the network component $Z_i$ in which it is contained. 
The simulation study presented below reveals that the probability density of shortest-path lengths in Manhattan grids exhibits properties that can not be observed in isotropic networks \cite{spl}.

In \cite{spl}, the typical shortest-path length was considered from the perspective of fixed-access networks, where it provides a cost indication for upgrading copper wires to fibers. In the setting of wireless networks, it reappears as a central performance indicator for D2D networks. Indeed, as users move predominantly along streets, relaying messages from a given user to its associated infrastructure component via other users is constrained by the topology and geometry of the street system. Hence, knowing the typical shortest-path length provides a strong indication on the number of required relaying hops, which can then be used to decide whether D2D is a viable option for the service under consideration.

To describe the probability density $f$ of the typical shortest-path length, we rely on the estimator $\wh f(r)$ from \cite[Theorem 2]{spl}. For $n \ge 1$ iid realizations $\C_1^*, \dots, \C_n^*$ of the typical cell on the typical street systems $S_1^*, \dots, S_n^*$, it is given as 
\begin{align}
	\label{denseEq}
	\wh f(r) = \frac\la n \sum_{i \le n} N_{i, r}
\end{align}
where 
$$N_{i, r} = \#\{y \in \C_i^* \cap S_i^*:\, \ell(y) = r\}$$ 
denotes the number of points on $\C_i^* \cap S_i^*$, whose shortest-path length to $o$ equals $r$. 

In the simulation study, we work with an inter-arrival distribution $\mc L(I)$ that is sufficiently flexible to capture both the mean as well as the fluctuations of distances. In order to achieve this, we choose $\mc L(I)$ as a truncated Gaussian. That is, $I \sim \mc N_+(\mu, \sigma^2)$ is distributed according to a normal random variable $\mc N(\mu, \sigma^2)$ with mean $\mu$ and variance $\sigma^2$ conditioned on being nonnegative. Unless stated otherwise we assume the same inter-arrival distribution for horizontal and vertical streets.

Based on the estimator \eqref{denseEq}, Figure~\ref{convFig} illustrates the density of the typical shortest-path length in a symmetric Manhattan grid.

%
%
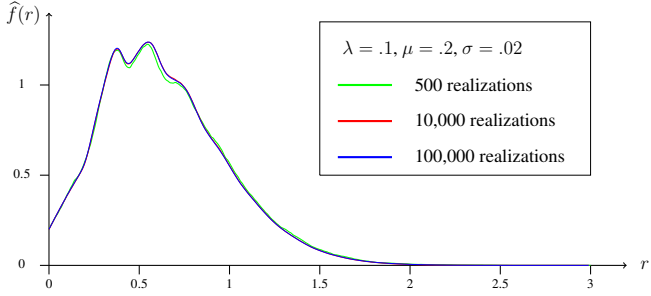
\begin{figure}[!htpb]
\centering
	\scalebox{.5}{\input{Convergence_speed_include}}
	\caption{Typical shortest-path length}
	\label{convFig}
\end{figure}

The plot reveals two striking features of typical shortest-path lengths based on Manhattan grids that do not occur in the isotropic models studied in~\cite[Section 5]{spl}. First, the density increases linearly for small distances. More precisely, we detect two phases of linear increase, where the slope in the second phase equals roughly three times the slope in the first phase. Moreover, for longer distances the density is bumpy exhibiting local minima and maxima. 

%
\subsection{Results for lattice-like system}
To understand better the origins of these peculiarities, we present in Figure~\ref{sigFig} the effects of changing the magnitude of variability in the street distances caused by different choices of $\sigma$. To fix ideas, we always think of the initial segment as horizontal. In the limiting regime $\sigma \to 0$ the streets arrange in a planar lattice with fixed mesh size $L$ corresponding to $\E I=2/\g\approx \mu$. Now, it becomes more clearly visible that the two initial phases mentioned above both occur over a segment of length $L$. Moreover, at the end of the second segment, the density is discontinuous and drops roughly by twice $\wh f(0)$. To simplify computations, we set $L = 1$ in the following, noting that the behavior for general $L$ is recovered via rescaling.

\begin{figure}[!htpb]
\centering
	\scalebox{.5}{\input{Sigma_include}}
	\caption{Effects of varying $\sigma$}
	\label{sigFig}
\end{figure}
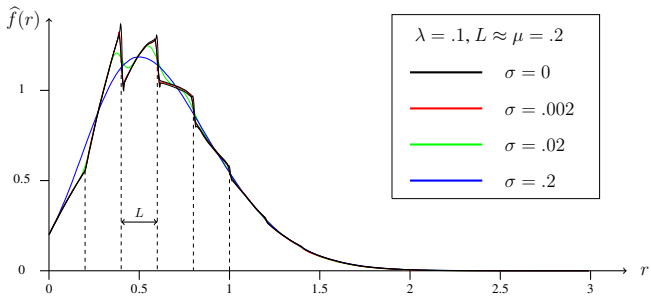

To elucidate this behavior, we first note that the ratio $ \g / \la$ of the street intensity divided by the base-station intensity equals roughly $100$, so that the typical cell is much larger than the mesh size. Hence, for small distances, variations in the typical shortest-path length are induced by the random location of the origin on the initial segment rather than the size of the confining cell. For $r$ close to 0, the segment of length $2r$ centered at $o$ is contained entirely in the initial segment with high probability, so that $N_{r} = 2$. Hence, $f(0) \approx 2 \la$. 

%
%
More precisely, for $r \le 1/2$, with probability $1 - 2r$ the origin lies far enough from the endpoints and $N_r = 2$. Conversely, with probability $2r$, it lies close to one of the end points. Then, Figure \ref{zSmallFig} shows that at the side of that point, we find three segments containing a point at distance $r$, so that in total $N_r = 4$. 
\begin{figure}[!htpb]
\centering
	\input{zSmall}
	\caption{$N_r = 4$ for $r$ small and $o$ close to a crossing}
	\label{zSmallFig}
\end{figure}
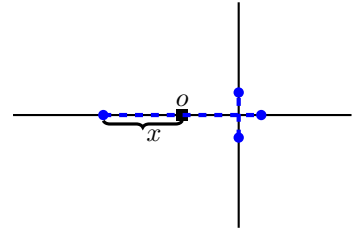

Hence, for $r \le 1/2$,
$$f(r) = \la \big(2(1 - 2r) + 8r\big) = {4\la r} + 2\la.$$
By a similar calculation, this relation remains valid for $1/2 \le r \le 1$. 

However, for $r \ge 1$ new effects appear. Indeed, for $1\le r \le \tfrac32 $, if $o$ is at distance at most $r - 1$ to one of the endpoints of the initial segment then $N_r = 12$,  and $N_r = 6$ otherwise. A detailed computation along these lines reveals that for $1 \le r < 2$, the density still increases linearly at a steeper slope
$$f(r) = 12\la r - 6 \la.$$

Finally, the discontinuity at $r = 2$ is caused by the occurrence of distance peaks. More precisely, Figure~\ref{zLargeFig} shows that for $r$ slightly smaller than 2, the two neighboring horizontal segments to the initial segment both contain two points at distance $r$. For $r > 2$ all four points disappear, so that the density drops by $4 \la = 2f(0)$. 

\begin{figure}[!htpb]
\centering
	\input{zLarge}
	\caption{Points at graph distance almost $2$ from $o$}
	\label{zLargeFig}
\end{figure}
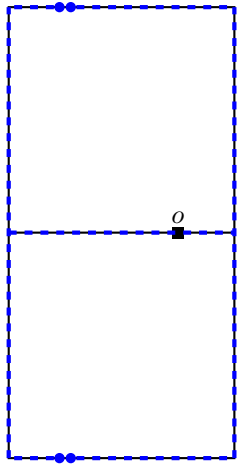

In principle, this analysis can be continued to yield a piecewise linear structure also for larger values of $r$. However, for finite $\g / \la$, at some stage the effects of the confining cell can no longer be neglected and induce the observed exponential decay of the density over long distances.

\subsection{Results for strong asymmetry between horizontal and vertical streets}

Finally, we consider the tail behavior of the typical shortest-path length. Conditioned on the street system, the network components are distributed as a Poisson point process, whose tails decay rapidly. Hence, from the theory side, one would expect that also the tails of the density $f(r)$ share this rapid decay, similarly as in the isotropic setting \cite{spl}. However, if there is a strong asymmetry between the horizontal and vertical street density, then this decay may become visible only at very long distances. 

Indeed, if the horizontal street density exceeds the vertical one substantially, then the origin is most likely located on a horizontal street. Moreover, the typical cell can contain several segments of parallel horizontal streets, but is unlikely to be hit by a vertical street. For points on the initial horizontal segment, the shortest-path length to the origin is small, as they can connect to it along a single line segment. However, for points on parallel segments inside the cell, the shortest paths to the origin are often very long. As illustrated in Figure~\ref{longDistFig}, such points first have to connect to one of the rare vertical streets, and after that again need to cover a substantial distance to reach back to the origin. 

\begin{figure}[!htpb]
\centering
	\input{longDist}
	\caption{Point on horizontal segment with long distance to $o$}
	\label{longDistFig}
\end{figure}
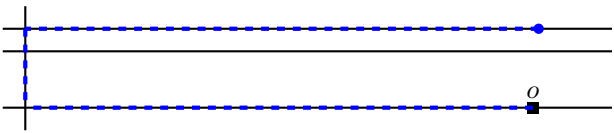

Finally, we show how the long path lengths manifest in the density $f(r)$. Indeed, Figure~\ref{ratioFig} illustrates that if horizontal and vertical streets occur at comparable intensities, then the density indeed decays rapidly for long distances. However, if there is a strong asymmetry between horizontal and vertical streets, then also distances that are far beyond the median of the distribution occur with a substantial probability.
%
%
\begin{figure}[!htpb]
\centering
	\scalebox{.5}{\input{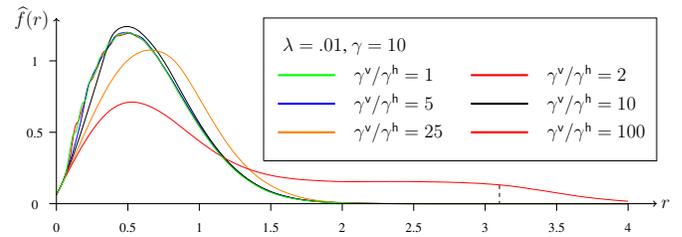}}
	\caption{Effects of asymmetry between horizontal and vertical street intensity}
	\label{ratioFig}
\end{figure}


\section{Conclusion}
\label{concSec}

This work paves the way to a detailed analysis of the typical cell in telecommunication networks with nodes located on anisotropic street systems. We have derived tractable representations of the Palm distributions of the (nested) Manhattan grid. Based on these descriptions we provided a simulation algorithm for the typical cell. Finally, regarding the implications of anisotropy on central network characteristics, we could illustrate in a simulation study that the probability density of typical shortest-path lengths can exhibit less regularity and stronger tails in comparison to the isotropic case.

These observations raise the question which other key network characteristics are affected by the anisotropy, regularity and long-range dependence in the Manhattan grid? Clearly, this includes the percolation probability and the stretch factor considered in \cite{d2dProceeding}. Moreover, to draw conclusions for more heterogeneous networks it is of interest to investigate how strongly effects of anisotropy prevail if some nodes are not on the streets, but scattered at random in the entire plane.

Moreover, the presented street-system model can serve as a starting point for deeper studies in the context of wireless communications, for example by considering mobile devices moving along streets. Also here it would be interesting to see, if anisotropy has a measurable effect on the quality of service of such systems.

\section*{Acknowledgements}
This research was supported by Orange S.A.~grant CRE H09261, by the Leibniz program \emph{Probabilistic Methods for Mobile Ad-Hoc Networks} and by The Danish Council for Independent Research -- Natural Sciences, grant DFF 7014-00074 Statistics for point processes in space and beyond, and by the Centre for Stochastic Geometry and Advanced Bioimaging, funded by grant 8721 from the Villum Foundation.

\bibliographystyle{IEEEtran}
\bibliography{../../wias}
\end{document}

%% file: ManhattanPic_1.tex
\begin{tikzpicture}[scale=0.7] 
 \begin{scope} 
\clip(2.0,2.0) rectangle (8.0,8.0);
\draw[red,very thick] (5.872407630121129,0)--(5.872407630121129,10);
\draw[red,very thick] (6.953990147856634,0)--(6.953990147856634,10);
\draw[red,very thick] (8.05432504777195,0)--(8.05432504777195,10);
\draw[red,very thick] (8.572699936797118,0)--(8.572699936797118,10);
\draw[red,very thick] (9.49103494245665,0)--(9.49103494245665,10);
\draw[red,very thick] (4.196742995158276,0)--(4.196742995158276,10);
\draw[red,very thick] (3.387204874804513,0)--(3.387204874804513,10);
\draw[red,very thick] (2.4432775057400704,0)--(2.4432775057400704,10);
\draw[red,very thick] (1.2552242720128504,0)--(1.2552242720128504,10);
\draw[red,very thick] (0.15897184222081542,0)--(0.15897184222081542,10);
\draw[red,very thick] (0,5.0)--(10,5.0);
\draw[red,very thick] (0,6.18298761838295)--(10,6.18298761838295);
\draw[red,very thick] (0,6.9453466153421335)--(10,6.9453466153421335);
\draw[red,very thick] (0,8.050189917261426)--(10,8.050189917261426);
\draw[red,very thick] (0,9.25540668296094)--(10,9.25540668296094);
\draw[red,very thick] (0,4.373108506364244)--(10,4.373108506364244);
\draw[red,very thick] (0,3.462816087343376)--(10,3.462816087343376);
\draw[red,very thick] (0,2.3291462447826223)--(10,2.3291462447826223);
\draw[red,very thick] (0,1.8721025517559386)--(10,1.8721025517559386);
\draw[red,very thick] (0,0.7431037623294224)--(10,0.7431037623294224);
\fill[blue] (5.0,5.0) circle (1pt);\draw(5.0,5.0) circle (1pt);
\fill[blue] (0.0,0.0) circle (1pt);\draw(0.0,0.0) circle (1pt);
\fill[blue] (5.87240763012,4.73897752708) circle (1pt);\draw(5.87240763012,4.73897752708) circle (1pt);
\fill[blue] (5.87240763012,8.12402760118) circle (1pt);\draw(5.87240763012,8.12402760118) circle (1pt);
\fill[blue] (5.87240763012,3.29552995367) circle (1pt);\draw(5.87240763012,3.29552995367) circle (1pt);
\fill[blue] (5.87240763012,4.51250827046) circle (1pt);\draw(5.87240763012,4.51250827046) circle (1pt);
\fill[blue] (6.95399014786,1.81054775856) circle (1pt);\draw(6.95399014786,1.81054775856) circle (1pt);
\fill[blue] (6.95399014786,9.86441267404) circle (1pt);\draw(6.95399014786,9.86441267404) circle (1pt);
\fill[blue] (6.95399014786,1.52155218887) circle (1pt);\draw(6.95399014786,1.52155218887) circle (1pt);
\fill[blue] (6.95399014786,5.64928116997) circle (1pt);\draw(6.95399014786,5.64928116997) circle (1pt);
\fill[blue] (8.05432504777,3.87637799364) circle (1pt);\draw(8.05432504777,3.87637799364) circle (1pt);
\fill[blue] (8.05432504777,7.7864263023) circle (1pt);\draw(8.05432504777,7.7864263023) circle (1pt);
\fill[blue] (8.05432504777,9.60276873523) circle (1pt);\draw(8.05432504777,9.60276873523) circle (1pt);
\fill[blue] (8.5726999368,2.77038725127) circle (1pt);\draw(8.5726999368,2.77038725127) circle (1pt);
\fill[blue] (8.5726999368,7.96602253003) circle (1pt);\draw(8.5726999368,7.96602253003) circle (1pt);
\fill[blue] (9.49103494246,7.6527039679) circle (1pt);\draw(9.49103494246,7.6527039679) circle (1pt);
\fill[blue] (9.49103494246,8.90657122907) circle (1pt);\draw(9.49103494246,8.90657122907) circle (1pt);
\fill[blue] (9.49103494246,0.60121533042) circle (1pt);\draw(9.49103494246,0.60121533042) circle (1pt);
\fill[blue] (4.19674299516,1.31130529776) circle (1pt);\draw(4.19674299516,1.31130529776) circle (1pt);
\fill[blue] (4.19674299516,6.98544359148) circle (1pt);\draw(4.19674299516,6.98544359148) circle (1pt);
\fill[blue] (4.19674299516,0.520510615762) circle (1pt);\draw(4.19674299516,0.520510615762) circle (1pt);
\fill[blue] (2.44327750574,6.60000088955) circle (1pt);\draw(2.44327750574,6.60000088955) circle (1pt);
\fill[blue] (2.44327750574,0.494992597249) circle (1pt);\draw(2.44327750574,0.494992597249) circle (1pt);
\fill[blue] (2.44327750574,5.61587040693) circle (1pt);\draw(2.44327750574,5.61587040693) circle (1pt);
\fill[blue] (2.44327750574,6.92241623091) circle (1pt);\draw(2.44327750574,6.92241623091) circle (1pt);
\fill[blue] (1.25522427201,2.80868730398) circle (1pt);\draw(1.25522427201,2.80868730398) circle (1pt);
\fill[blue] (1.25522427201,2.77559350578) circle (1pt);\draw(1.25522427201,2.77559350578) circle (1pt);
\fill[blue] (1.25522427201,8.8564338756) circle (1pt);\draw(1.25522427201,8.8564338756) circle (1pt);
\fill[blue] (0.158971842221,0.461740751541) circle (1pt);\draw(0.158971842221,0.461740751541) circle (1pt);
\fill[blue] (0.158971842221,5.98218645789) circle (1pt);\draw(0.158971842221,5.98218645789) circle (1pt);
\fill[blue] (0.158971842221,4.78750088889) circle (1pt);\draw(0.158971842221,4.78750088889) circle (1pt);
\fill[blue] (0.158971842221,2.22358684319) circle (1pt);\draw(0.158971842221,2.22358684319) circle (1pt);
\fill[blue] (0.158971842221,9.8848077957) circle (1pt);\draw(0.158971842221,9.8848077957) circle (1pt);
\fill[blue] (6.72218382057,5.0) circle (1pt);\draw(6.72218382057,5.0) circle (1pt);
\fill[blue] (1.92778804652,5.0) circle (1pt);\draw(1.92778804652,5.0) circle (1pt);
\fill[blue] (0.945707899269,5.0) circle (1pt);\draw(0.945707899269,5.0) circle (1pt);
\fill[blue] (3.50458634591,5.0) circle (1pt);\draw(3.50458634591,5.0) circle (1pt);
\fill[blue] (7.30005195489,5.0) circle (1pt);\draw(7.30005195489,5.0) circle (1pt);
\fill[blue] (8.40845158554,8.05018991726) circle (1pt);\draw(8.40845158554,8.05018991726) circle (1pt);
\fill[blue] (6.40778051033,8.05018991726) circle (1pt);\draw(6.40778051033,8.05018991726) circle (1pt);
\fill[blue] (9.53252974647,8.05018991726) circle (1pt);\draw(9.53252974647,8.05018991726) circle (1pt);
\fill[blue] (8.3078108785,9.25540668296) circle (1pt);\draw(8.3078108785,9.25540668296) circle (1pt);
\fill[blue] (3.22189493451,9.25540668296) circle (1pt);\draw(3.22189493451,9.25540668296) circle (1pt);
\fill[blue] (1.54273634343,9.25540668296) circle (1pt);\draw(1.54273634343,9.25540668296) circle (1pt);
\fill[blue] (6.12806235037,9.25540668296) circle (1pt);\draw(6.12806235037,9.25540668296) circle (1pt);
\fill[blue] (8.2568715456,9.25540668296) circle (1pt);\draw(8.2568715456,9.25540668296) circle (1pt);
\fill[blue] (7.99392382552,4.37310850636) circle (1pt);\draw(7.99392382552,4.37310850636) circle (1pt);
\fill[blue] (5.77901750929,4.37310850636) circle (1pt);\draw(5.77901750929,4.37310850636) circle (1pt);
\fill[blue] (8.97033045409,4.37310850636) circle (1pt);\draw(8.97033045409,4.37310850636) circle (1pt);
\fill[blue] (3.97963600866,4.37310850636) circle (1pt);\draw(3.97963600866,4.37310850636) circle (1pt);
\fill[blue] (7.54675076512,4.37310850636) circle (1pt);\draw(7.54675076512,4.37310850636) circle (1pt);
\fill[blue] (0.18128847727,4.37310850636) circle (1pt);\draw(0.18128847727,4.37310850636) circle (1pt);
\fill[blue] (7.43828010617,3.46281608734) circle (1pt);\draw(7.43828010617,3.46281608734) circle (1pt);
\fill[blue] (6.76206718349,3.46281608734) circle (1pt);\draw(6.76206718349,3.46281608734) circle (1pt);
\fill[blue] (0.479153163448,2.32914624478) circle (1pt);\draw(0.479153163448,2.32914624478) circle (1pt);
\fill[blue] (6.9193270481,2.32914624478) circle (1pt);\draw(6.9193270481,2.32914624478) circle (1pt);
\fill[blue] (2.52666521426,2.32914624478) circle (1pt);\draw(2.52666521426,2.32914624478) circle (1pt);
\fill[blue] (7.96969442409,2.32914624478) circle (1pt);\draw(7.96969442409,2.32914624478) circle (1pt);
\fill[blue] (7.51517705868,2.32914624478) circle (1pt);\draw(7.51517705868,2.32914624478) circle (1pt);
\fill[blue] (2.78777383328,2.32914624478) circle (1pt);\draw(2.78777383328,2.32914624478) circle (1pt);
\fill[blue] (5.67235568721,2.32914624478) circle (1pt);\draw(5.67235568721,2.32914624478) circle (1pt);
\fill[blue] (8.31124678991,2.32914624478) circle (1pt);\draw(8.31124678991,2.32914624478) circle (1pt);
\fill[blue] (6.01281656411,1.87210255176) circle (1pt);\draw(6.01281656411,1.87210255176) circle (1pt);
\fill[blue] (7.70107621387,1.87210255176) circle (1pt);\draw(7.70107621387,1.87210255176) circle (1pt);
\fill[blue] (4.39462636555,1.87210255176) circle (1pt);\draw(4.39462636555,1.87210255176) circle (1pt);
\fill[blue] (4.27329500179,0.743103762329) circle (1pt);\draw(4.27329500179,0.743103762329) circle (1pt);
\fill[blue] (1.98704663481,0.743103762329) circle (1pt);\draw(1.98704663481,0.743103762329) circle (1pt);
\fill[blue] (5.28877472857,0.743103762329) circle (1pt);\draw(5.28877472857,0.743103762329) circle (1pt);
\fill[blue] (3.19260947879,0.743103762329) circle (1pt);\draw(3.19260947879,0.743103762329) circle (1pt);
\draw[black,thin] (0.419694527127,5.38484367339)--(0.630308469407,4.60508708221);
\draw[black,thin] (0.626891525535,1.34266379736)--(0.968080834976,1.2841642137);
\draw[black,thin] (45.1892943363,4.12695964916)--(21.1090935989,4.12695964916);
\draw[black,thin] (21.1090935989,4.12695964916)--(11.7444174105,2.83417942025);
\draw[black,thin] (7.96229087588,0.344044205746)--(8.62391613545,1.27589865243);
\draw[black,thin] (8.62391613545,1.27589865243)--(9.53824644714,1.9001805958);
\draw[black,thin] (11.7444174105,2.83417942025)--(9.53824644714,1.9001805958);
\draw[black,thin] (6.47193997905,1.66604997372)--(6.20584999782,0.951639470251);
\draw[black,thin] (6.47193997905,1.66604997372)--(6.49642470749,2.04042178801);
\draw[black,thin] (6.49642470749,2.04042178801)--(6.29584136765,2.43826318314);
\draw[black,thin] (6.29584136765,2.43826318314)--(5.20372146483,1.62472149016);
\draw[black,thin] (6.20584999782,0.951639470251)--(5.20372146483,1.59431773463);
\draw[black,thin] (5.20372146483,1.62472149016)--(5.20372146483,1.59431773463);
\draw[black,thin] (7.96229087588,0.344044205746)--(7.34197465961,1.66604997372);
\draw[black,thin] (7.34197465961,1.66604997372)--(6.47193997905,1.66604997372);
\draw[black,thin] (6.29584136765,2.43826318314)--(6.29584136765,2.70397624373);
\draw[black,thin] (5.20372146483,1.62472149016)--(4.6994639073,3.03444374812);
\draw[black,thin] (4.6994639073,3.03444374812)--(6.29584136765,2.70397624373);
\draw[black,thin] (1.61123569621,1.75578410748)--(2.5898280568,1.42283803237);
\draw[black,thin] (1.61123569621,1.75578410748)--(1.37447743536,1.67053644981);
\draw[black,thin] (2.5898280568,1.42283803237)--(2.5898280568,1.30799012758);
\draw[black,thin] (1.37447743536,1.67053644981)--(0.968080834976,1.2841642137);
\draw[black,thin] (1.19413041096,3.91757345816)--(1.43674797289,3.95184246407);
\draw[black,thin] (1.19413041096,3.91757345816)--(0.379423451513,3.35829762845);
\draw[black,thin] (0.729256903201,2.79214040488)--(0.379423451513,3.35829762845);
\draw[black,thin] (0.729256903201,2.79214040488)--(1.97513653896,2.79214040488);
\draw[black,thin] (1.43674797289,3.95184246407)--(2.31129604266,3.68342374325);
\draw[black,thin] (1.97513653896,2.79214040488)--(2.31129604266,3.68342374325);
\draw[black,thin] (0.630308469407,4.60508708221)--(1.19413041096,3.91757345816);
\draw[black,thin] (1.37447743536,1.67053644981)--(0.729256903201,2.79214040488);
\draw[black,thin] (1.61123569621,1.75578410748)--(1.97513653896,2.79214040488);
\draw[black,thin] (2.65721952377,3.76098908138)--(2.65721952377,1.4511342233);
\draw[black,thin] (2.65721952377,3.76098908138)--(2.31129604266,3.68342374325);
\draw[black,thin] (2.5898280568,1.42283803237)--(2.65721952377,1.4511342233);
\draw[black,thin] (5.70024704514,6.5750603408)--(4.43766326895,8.43321658084);
\draw[black,thin] (5.70024704514,6.5750603408)--(5.94639185493,6.68263739182);
\draw[black,thin] (5.94639185493,6.68263739182)--(6.21646591568,8.64085601529);
\draw[black,thin] (6.21646591568,8.64085601529)--(4.67497864244,8.98918179217);
\draw[black,thin] (4.67497864244,8.98918179217)--(4.43766326895,8.43321658084);
\draw[black,thin] (7.07224676038,8.83947376608)--(6.21646591568,8.64085601529);
\draw[black,thin] (7.07224676038,8.83947376608)--(4.67497864244,12.0906246671);
\draw[black,thin] (4.67497864244,12.0906246671)--(4.67497864244,8.98918179217);
\draw[black,thin] (7.07224676038,8.83947376608)--(7.24444414357,8.7876301488);
\draw[black,thin] (3.58321381071,18.7811504785)--(4.67497864244,12.0906246671);
\draw[black,thin] (7.31523214057,8.93907096519)--(7.24444414357,8.7876301488);
\draw[black,thin] (2.38231563897,12.9370890886)--(2.38231563897,8.3472986069);
\draw[black,thin] (2.38231563897,12.9370890886)--(0.810023903588,9.48034037158);
\draw[black,thin] (0.810023903588,9.48034037158)--(2.38231563897,8.3472986069);
\draw[black,thin] (3.58321381071,18.7811504785)--(2.38231563897,12.9370890886);
\draw[black,thin] (7.34197465961,1.66604997372)--(7.31596001926,1.98178777913);
\draw[black,thin] (6.49642470749,2.04042178801)--(7.21725205339,2.08860185575);
\draw[black,thin] (7.31596001926,1.98178777913)--(7.21725205339,2.08860185575);
\draw[black,thin] (3.86676648234,3.06944761357)--(3.5239058384,1.86403442807);
\draw[black,thin] (3.86676648234,3.06944761357)--(4.55564633993,3.18375301834);
\draw[black,thin] (5.20372146483,1.59431773463)--(4.9215657066,1.37085504211);
\draw[black,thin] (4.55564633993,3.18375301834)--(4.6994639073,3.03444374812);
\draw[black,thin] (4.9215657066,1.37085504211)--(3.5239058384,1.86403442807);
\draw[black,thin] (2.67242421111,3.76588482492)--(2.65721952377,3.76098908138);
\draw[black,thin] (2.67242421111,3.76588482492)--(3.86676648234,3.06944761357);
\draw[black,thin] (2.65721952377,1.4511342233)--(3.35413129371,1.62902020641);
\draw[black,thin] (3.35413129371,1.62902020641)--(3.5239058384,1.86403442807);
\draw[black,thin] (2.67242421111,3.76588482492)--(2.71618719622,3.90912331631);
\draw[black,thin] (2.71618719622,3.90912331631)--(4.25229317296,5.07316308014);
\draw[black,thin] (4.55564633993,3.18375301834)--(4.87932675898,3.75229914434);
\draw[black,thin] (4.25229317296,5.07316308014)--(4.87932675898,4.05256781121);
\draw[black,thin] (4.87932675898,3.75229914434)--(4.87932675898,4.05256781121);
\draw[black,thin] (6.29584136765,2.70397624373)--(6.41909084968,2.83749699207);
\draw[black,thin] (4.87932675898,3.75229914434)--(6.22514524022,3.86893672056);
\draw[black,thin] (6.41909084968,2.83749699207)--(6.22514524022,3.86893672056);
\draw[black,thin] (6.07365779957,5.59755705008)--(6.37216620195,4.62574289877);
\draw[black,thin] (6.07365779957,5.59755705008)--(5.76408974163,5.9653721216);
\draw[black,thin] (5.76408974163,5.9653721216)--(5.37671252493,4.67065279798);
\draw[black,thin] (6.37216620195,4.62574289877)--(5.55265341353,4.62574289877);
\draw[black,thin] (5.55265341353,4.62574289877)--(5.37671252493,4.67065279798);
\draw[black,thin] (5.61793799926,6.4052109432)--(4.25229317296,5.85270782623);
\draw[black,thin] (5.61793799926,6.4052109432)--(5.76408974163,5.9653721216);
\draw[black,thin] (4.25229317296,5.85270782623)--(4.25229317296,5.07316308014);
\draw[black,thin] (4.87932675898,4.05256781121)--(5.37671252493,4.67065279798);
\draw[black,thin] (6.60049267365,4.22773327616)--(6.40017170517,4.05795258596);
\draw[black,thin] (6.60049267365,4.22773327616)--(6.37216620195,4.62574289877);
\draw[black,thin] (6.40017170517,4.05795258596)--(5.55265341353,4.62574289877);
\draw[black,thin] (6.40017170517,4.05795258596)--(6.22514524022,3.86893672056);
\draw[black,thin] (2.71618719622,4.86377234556)--(3.43816682223,6.10793564824);
\draw[black,thin] (2.71618719622,4.86377234556)--(1.43674797289,5.93467529188);
\draw[black,thin] (3.43816682223,6.10793564824)--(1.36549524066,6.10793564824);
\draw[black,thin] (1.43674797289,5.93467529188)--(1.36549524066,6.10793564824);
\draw[black,thin] (2.71618719622,3.90912331631)--(2.71618719622,4.86377234556);
\draw[black,thin] (4.25229317296,5.85270782623)--(3.46641902168,6.12667583097);
\draw[black,thin] (3.46641902168,6.12667583097)--(3.43816682223,6.10793564824);
\draw[black,thin] (1.43674797289,3.95184246407)--(1.43674797289,5.93467529188);
\draw[black,thin] (0.419694527127,5.38484367339)--(1.34563385943,6.12652551344);
\draw[black,thin] (1.34563385943,6.12652551344)--(1.36549524066,6.10793564824);
\draw[black,thin] (3.28463683951,7.93804307655)--(2.40584096999,8.23133436703);
\draw[black,thin] (3.28463683951,7.93804307655)--(3.32693751429,6.76120856023);
\draw[black,thin] (3.32693751429,6.76120856023)--(1.17397717088,6.76120856023);
\draw[black,thin] (2.40584096999,8.23133436703)--(0.939599749092,7.33063285136);
\draw[black,thin] (0.939599749092,7.33063285136)--(1.17397717088,6.76120856023);
\draw[black,thin] (4.43766326895,8.43321658084)--(3.28463683951,7.93804307655);
\draw[black,thin] (2.38231563897,8.3472986069)--(2.40584096999,8.23133436703);
\draw[black,thin] (5.61793799926,6.4052109432)--(5.70024704514,6.5750603408);
\draw[black,thin] (3.46641902168,6.12667583097)--(3.32693751429,6.76120856023);
\draw[black,thin] (1.34563385943,6.12652551344)--(1.17397717088,6.76120856023);
\draw[black,thin] (7.31523214057,8.93907096519)--(8.28234121205,9.50299138816);
\draw[black,thin] (9.22709221545,10.192419008)--(8.28234121205,9.50299138816);
\draw[black,thin] (5.94639185493,6.68263739182)--(7.07421007001,6.93921739685);
\draw[black,thin] (7.24444414357,8.7876301488)--(7.3454795946,8.63261758006);
\draw[black,thin] (7.3454795946,8.63261758006)--(7.07421007001,6.93921739685);
\draw[black,thin] (8.63514190265,6.12845725356)--(8.62569990323,6.14041417824);
\draw[black,thin] (8.63514190265,6.12845725356)--(7.01111788773,5.26286510141);
\draw[black,thin] (7.01111788773,5.26286510141)--(6.07365779957,5.59755705008);
\draw[black,thin] (8.62569990323,6.14041417824)--(7.07421007001,6.93921739685);
\draw[black,thin] (8.62391613545,1.27589865243)--(8.140470607,1.92131667499);
\draw[black,thin] (7.31596001926,1.98178777913)--(7.74243574138,2.15525363735);
\draw[black,thin] (8.140470607,1.92131667499)--(7.74243574138,2.15525363735);
\draw[black,thin] (9.53824644714,1.9001805958)--(8.140470607,2.72841932856);
\draw[black,thin] (8.140470607,1.92131667499)--(8.140470607,2.72841932856);
\draw[black,thin] (4.78103486518,0.444026878231)--(4.78103486518,1.10076753127);
\draw[black,thin] (4.78103486518,0.444026878231)--(3.73295224029,0.804472962424);
\draw[black,thin] (4.78103486518,1.10076753127)--(3.73295224029,0.959562654046);
\draw[black,thin] (3.73295224029,0.804472962424)--(3.73295224029,0.959562654046);
\draw[black,thin] (4.9215657066,1.37085504211)--(4.78103486518,1.10076753127);
\draw[black,thin] (3.35413129371,1.62902020641)--(3.73295224029,0.959562654046);
\draw[black,thin] (8.4821271398,5.61092416287)--(8.66661269021,6.10246416841);
\draw[black,thin] (8.4821271398,5.61092416287)--(8.4821271398,4.18043526481);
\draw[black,thin] (8.66661269021,6.10246416841)--(21.1090935989,4.12695964916);
\draw[black,thin] (11.7444174105,2.83417942025)--(8.81853093504,3.56008340953);
\draw[black,thin] (8.4821271398,4.18043526481)--(8.81853093504,3.56008340953);
\draw[black,thin] (8.4821271398,5.61092416287)--(7.77033729532,4.82308295606);
\draw[black,thin] (8.66661269021,6.10246416841)--(8.63514190265,6.12845725356);
\draw[black,thin] (7.01111788773,5.26286510141)--(7.01111788773,4.52430952164);
\draw[black,thin] (7.77033729532,4.82308295606)--(7.01111788773,4.52430952164);
\draw[black,thin] (8.4821271398,4.18043526481)--(7.77033729532,4.09388335085);
\draw[black,thin] (7.77033729532,4.82308295606)--(7.77033729532,4.09388335085);
\draw[black,thin] (6.60049267365,4.22773327616)--(6.78936625451,4.23263373961);
\draw[black,thin] (7.01111788773,4.52430952164)--(6.78936625451,4.23263373961);
\draw[black,thin] (9.22709221545,10.192419008)--(8.78366061398,8.68833180668);
\draw[black,thin] (8.28234121205,9.50299138816)--(8.28234121205,8.64646951228);
\draw[black,thin] (8.28234121205,8.64646951228)--(8.78366061398,8.68833180668);
\draw[black,thin] (7.3454795946,8.63261758006)--(7.73985162431,8.57824062955);
\draw[black,thin] (8.28234121205,8.64646951228)--(7.73985162431,8.57824062955);
\draw[black,thin] (7.10017364483,3.96471383702)--(7.5871400049,3.9066868109);
\draw[black,thin] (7.10017364483,3.96471383702)--(7.10017364483,2.93197510565);
\draw[black,thin] (7.5871400049,3.9066868109)--(8.05885053111,3.20402324556);
\draw[black,thin] (7.21725205339,2.87838084401)--(7.74243574138,2.91400411244);
\draw[black,thin] (7.21725205339,2.87838084401)--(7.10017364483,2.93197510565);
\draw[black,thin] (7.74243574138,2.91400411244)--(7.9376960908,3.00553353442);
\draw[black,thin] (7.9376960908,3.00553353442)--(8.05885053111,3.20402324556);
\draw[black,thin] (7.77033729532,4.09388335085)--(7.5871400049,3.9066868109);
\draw[black,thin] (6.78936625451,4.23263373961)--(7.10017364483,3.96471383702);
\draw[black,thin] (6.41909084968,2.83749699207)--(7.10017364483,2.93197510565);
\draw[black,thin] (8.81853093504,3.56008340953)--(8.05885053111,3.20402324556);
\draw[black,thin] (7.21725205339,2.08860185575)--(7.21725205339,2.87838084401);
\draw[black,thin] (7.74243574138,2.15525363735)--(7.74243574138,2.91400411244);
\draw[black,thin] (8.140470607,2.72841932856)--(7.9376960908,3.00553353442);
\draw[black,thin] (9.01349370426,8.45423666766)--(8.81733580913,8.64576174438);
\draw[black,thin] (9.01349370426,8.45423666766)--(9.06223704097,7.89837638573);
\draw[black,thin] (8.81733580913,8.64576174438)--(8.35767488342,7.74875695823);
\draw[black,thin] (9.06223704097,7.89837638573)--(8.68694081013,6.79838516934);
\draw[black,thin] (8.68694081013,6.79838516934)--(8.35767488342,7.74875695823);
\draw[black,thin] (8.78366061398,8.68833180668)--(8.81733580913,8.64576174438);
\draw[black,thin] (45.1892943363,4.12695964916)--(9.06223704097,7.89837638573);
\draw[black,thin] (7.73985162431,8.57824062955)--(8.35767488342,7.74875695823);
\draw[black,thin] (8.62569990323,6.14041417824)--(8.68694081013,6.79838516934);
\fill[green] (5.0,5.0) circle (1.5pt);\draw(5.0,5.0) circle (1.5pt);
\draw[black,ultra thick] (5.37671252493,4.67065279798)--(5.76408974163,5.9653721216) -- (5.61793799926,6.4052109432) --(4.25229317296,5.85270782623)--(4.25229317296,5.07316308014)--(4.87932675898,4.05256781121)-- cycle;
 \end{scope} 
 \draw (2.0,2.0) rectangle (8.0,8.0); 
 \end{tikzpicture}

%% file: NestedManhattanPic.tex
\begin{tikzpicture}[scale=0.7] 
 \begin{scope} 
\clip(2.0,2.0) rectangle (8.0,8.0);
\draw[red,very thick] (6.046889156145355,0)--(6.046889156145355,10);
\draw[red,very thick] (7.344788177427959,0)--(7.344788177427959,10);
\draw[red,very thick] (8.66519005732634,0)--(8.66519005732634,10);
\draw[red,very thick] (9.287239924156541,0)--(9.287239924156541,10);
\draw[red,very thick] (3.9744854166354164,0)--(3.9744854166354164,10);
\draw[red,very thick] (3.0105770108253482,0)--(3.0105770108253482,10);
\draw[red,very thick] (2.0391312664008323,0)--(2.0391312664008323,10);
\draw[red,very thick] (0.9064184235235009,0)--(0.9064184235235009,10);
\draw[red,very thick] (0,5.0)--(10,5.0);
\draw[red,very thick] (0,6.315502915750441)--(10,6.315502915750441);
\draw[red,very thick] (0,7.759944668133779)--(10,7.759944668133779);
\draw[red,very thick] (0,9.17952981019332)--(10,9.17952981019332);
\draw[red,very thick] (0,3.674188037696849)--(10,3.674188037696849);
\draw[red,very thick] (0,2.2279279188574326)--(10,2.2279279188574326);
\draw[red,very thick] (0,1.0998818656220464)--(10,1.0998818656220464);
\draw[red,very thick] (0,0.3476120732591395)--(10,0.3476120732591395);
\draw[red,thin] (1.20695486836,0.3476120732591395)--(1.20695486836,1.0998818656220464);
\draw[red,thin] (1.65605022665,0.3476120732591395)--(1.65605022665,1.0998818656220464);
\draw[red,thin] (0.9064184235235009,0.684698023181)--(2.0391312664008323,0.684698023181);
\draw[red,thin] (0.9064184235235009,0.87482559871)--(2.0391312664008323,0.87482559871);
\draw[red,thin] (1.29584353569,1.0998818656220464)--(1.29584353569,2.2279279188574326);
\draw[red,thin] (0.9064184235235009,1.5401278282)--(2.0391312664008323,1.5401278282);
\draw[red,thin] (1.71772186696,2.2279279188574326)--(1.71772186696,3.674188037696849);
\draw[red,thin] (0.9064184235235009,2.62021903287)--(2.0391312664008323,2.62021903287);
\draw[red,thin] (1.41532545478,3.674188037696849)--(1.41532545478,5.0);
\draw[red,thin] (1.63704534811,3.674188037696849)--(1.63704534811,5.0);
\draw[red,thin] (0.9064184235235009,3.91929181996)--(2.0391312664008323,3.91929181996);
\draw[red,thin] (0.9064184235235009,4.22968009612)--(2.0391312664008323,4.22968009612);
\draw[red,thin] (1.1787820092,5.0)--(1.1787820092,6.315502915750441);
\draw[red,thin] (1.48777314347,5.0)--(1.48777314347,6.315502915750441);
\draw[red,thin] (0.9064184235235009,5.29363229656)--(2.0391312664008323,5.29363229656);
\draw[red,thin] (1.59905039472,6.315502915750441)--(1.59905039472,7.759944668133779);
\draw[red,thin] (0.9064184235235009,6.58007535429)--(2.0391312664008323,6.58007535429);
\draw[red,thin] (1.11447783792,7.759944668133779)--(1.11447783792,9.17952981019332);
\draw[red,thin] (1.36214925989,7.759944668133779)--(1.36214925989,9.17952981019332);
\draw[red,thin] (1.88051705289,7.759944668133779)--(1.88051705289,9.17952981019332);
\draw[red,thin] (0.9064184235235009,8.15156880545)--(2.0391312664008323,8.15156880545);
\draw[red,thin] (2.39078222499,0.3476120732591395)--(2.39078222499,1.0998818656220464);
\draw[red,thin] (2.40732787673,0.3476120732591395)--(2.40732787673,1.0998818656220464);
\draw[red,thin] (2.88338870641,0.3476120732591395)--(2.88338870641,1.0998818656220464);
\draw[red,thin] (2.0391312664008323,0.841846223609)--(3.0105770108253482,0.841846223609);
\draw[red,thin] (2.0391312664008323,1.03126640948)--(3.0105770108253482,1.03126640948);
\draw[red,thin] (2.44713843585,1.0998818656220464)--(2.44713843585,2.2279279188574326);
\draw[red,thin] (2.0391312664008323,1.91799768326)--(3.0105770108253482,1.91799768326);
\draw[red,thin] (2.61751359316,2.2279279188574326)--(2.61751359316,3.674188037696849);
\draw[red,thin] (2.0391312664008323,2.56116657638)--(3.0105770108253482,2.56116657638);
\draw[red,thin] (2.27736113493,3.674188037696849)--(2.27736113493,5.0);
\draw[red,thin] (2.54126042068,3.674188037696849)--(2.54126042068,5.0);
\draw[red,thin] (2.72124996687,3.674188037696849)--(2.72124996687,5.0);
\draw[red,thin] (2.0391312664008323,4.16172913349)--(3.0105770108253482,4.16172913349);
\draw[red,thin] (2.0391312664008323,4.569781093)--(3.0105770108253482,4.569781093);
\draw[red,thin] (2.34942107116,5.0)--(2.34942107116,6.315502915750441);
\draw[red,thin] (2.40622099498,5.0)--(2.40622099498,6.315502915750441);
\draw[red,thin] (2.78308994904,5.0)--(2.78308994904,6.315502915750441);
\draw[red,thin] (2.0391312664008323,5.56005953129)--(3.0105770108253482,5.56005953129);
\draw[red,thin] (2.30717957642,6.315502915750441)--(2.30717957642,7.759944668133779);
\draw[red,thin] (2.78624467029,6.315502915750441)--(2.78624467029,7.759944668133779);
\draw[red,thin] (2.0391312664008323,6.84396587359)--(3.0105770108253482,6.84396587359);
\draw[red,thin] (2.0391312664008323,7.12376346692)--(3.0105770108253482,7.12376346692);
\draw[red,thin] (2.30585266707,7.759944668133779)--(2.30585266707,9.17952981019332);
\draw[red,thin] (2.0391312664008323,8.30287107193)--(3.0105770108253482,8.30287107193);
\draw[red,thin] (3.42468264036,0.3476120732591395)--(3.42468264036,1.0998818656220464);
\draw[red,thin] (3.0105770108253482,0.545296677725)--(3.9744854166354164,0.545296677725);
\draw[red,thin] (3.0105770108253482,1.03413825048)--(3.9744854166354164,1.03413825048);
\draw[red,thin] (3.0105770108253482,1.59179649387)--(3.9744854166354164,1.59179649387);
\draw[red,thin] (3.22559892241,1.0998818656220464)--(3.22559892241,2.2279279188574326);
\draw[red,thin] (3.71031362867,1.0998818656220464)--(3.71031362867,2.2279279188574326);
\draw[red,thin] (3.0105770108253482,1.7486386568)--(3.9744854166354164,1.7486386568);
\draw[red,thin] (3.0105770108253482,2.00967289189)--(3.9744854166354164,2.00967289189);
\draw[red,thin] (3.46274462798,2.2279279188574326)--(3.46274462798,3.674188037696849);
\draw[red,thin] (3.0105770108253482,2.73908846653)--(3.9744854166354164,2.73908846653);
\draw[red,thin] (3.0105770108253482,3.03511889919)--(3.9744854166354164,3.03511889919);
\draw[red,thin] (3.0105770108253482,3.18469603557)--(3.9744854166354164,3.18469603557);
\draw[red,thin] (3.3011794688,3.674188037696849)--(3.3011794688,5.0);
\draw[red,thin] (3.0105770108253482,4.47110368449)--(3.9744854166354164,4.47110368449);
\draw[red,thin] (3.55798286783,5.0)--(3.55798286783,6.315502915750441);
\draw[red,thin] (3.0105770108253482,5.29564550188)--(3.9744854166354164,5.29564550188);
\draw[red,thin] (3.0105770108253482,5.84662717002)--(3.9744854166354164,5.84662717002);
\draw[red,thin] (3.40843686529,6.315502915750441)--(3.40843686529,7.759944668133779);
\draw[red,thin] (3.0105770108253482,6.80305097912)--(3.9744854166354164,6.80305097912);
\draw[red,thin] (3.0105770108253482,7.1493682221)--(3.9744854166354164,7.1493682221);
\draw[red,thin] (3.0105770108253482,7.57832922361)--(3.9744854166354164,7.57832922361);
\draw[red,thin] (3.41134875244,7.759944668133779)--(3.41134875244,9.17952981019332);
\draw[red,thin] (3.55776550076,7.759944668133779)--(3.55776550076,9.17952981019332);
\draw[red,thin] (3.0105770108253482,8.12077803202)--(3.9744854166354164,8.12077803202);
\draw[red,thin] (3.0105770108253482,8.83076656332)--(3.9744854166354164,8.83076656332);
\draw[red,thin] (4.97875172001,0.3476120732591395)--(4.97875172001,1.0998818656220464);
\draw[red,thin] (5.31478549115,0.3476120732591395)--(5.31478549115,1.0998818656220464);
\draw[red,thin] (3.9744854166354164,0.943034712192)--(6.046889156145355,0.943034712192);
\draw[red,thin] (4.54561366039,1.0998818656220464)--(4.54561366039,2.2279279188574326);
\draw[red,thin] (5.59388078183,1.0998818656220464)--(5.59388078183,2.2279279188574326);
\draw[red,thin] (3.9744854166354164,1.62170655577)--(6.046889156145355,1.62170655577);
\draw[red,thin] (3.9744854166354164,1.9775782321)--(6.046889156145355,1.9775782321);
\draw[red,thin] (4.36217479517,2.2279279188574326)--(4.36217479517,3.674188037696849);
\draw[red,thin] (5.1347021412,2.2279279188574326)--(5.1347021412,3.674188037696849);
\draw[red,thin] (5.34023069859,2.2279279188574326)--(5.34023069859,3.674188037696849);
\draw[red,thin] (3.9744854166354164,2.8752500734)--(6.046889156145355,2.8752500734);
\draw[red,thin] (3.9744854166354164,2.99868958468)--(6.046889156145355,2.99868958468);
\draw[red,thin] (5.11887259329,3.674188037696849)--(5.11887259329,5.0);
\draw[red,thin] (3.9744854166354164,4.18162248239)--(6.046889156145355,4.18162248239);
\draw[red,thin] (4.53466611235,5.0)--(4.53466611235,6.315502915750441);
\draw[red,thin] (4.87008093913,5.0)--(4.87008093913,6.315502915750441);
\draw[red,thin] (5.41518002552,5.0)--(5.41518002552,6.315502915750441);
\draw[red,thin] (3.9744854166354164,5.49657922038)--(6.046889156145355,5.49657922038);
\draw[red,thin] (3.9744854166354164,5.78644322103)--(6.046889156145355,5.78644322103);
\draw[red,thin] (3.9744854166354164,5.99906997584)--(6.046889156145355,5.99906997584);
\draw[red,thin] (4.82233174023,6.315502915750441)--(4.82233174023,7.759944668133779);
\draw[red,thin] (3.9744854166354164,6.70385841154)--(6.046889156145355,6.70385841154);
\draw[red,thin] (3.9744854166354164,7.06647456548)--(6.046889156145355,7.06647456548);
\draw[red,thin] (4.67864052529,7.759944668133779)--(4.67864052529,9.17952981019332);
\draw[red,thin] (5.43190071776,7.759944668133779)--(5.43190071776,9.17952981019332);
\draw[red,thin] (3.9744854166354164,8.46536412978)--(6.046889156145355,8.46536412978);
\draw[red,thin] (6.96718958147,0.3476120732591395)--(6.96718958147,1.0998818656220464);
\draw[red,thin] (6.046889156145355,1.19984654886)--(7.344788177427959,1.19984654886);
\draw[red,thin] (6.46671080491,1.0998818656220464)--(6.46671080491,2.2279279188574326);
\draw[red,thin] (6.046889156145355,1.69250207506)--(7.344788177427959,1.69250207506);
\draw[red,thin] (6.64846606123,2.2279279188574326)--(6.64846606123,3.674188037696849);
\draw[red,thin] (7.09529125953,2.2279279188574326)--(7.09529125953,3.674188037696849);
\draw[red,thin] (6.046889156145355,2.64537339189)--(7.344788177427959,2.64537339189);
\draw[red,thin] (6.046889156145355,3.19415932836)--(7.344788177427959,3.19415932836);
\draw[red,thin] (6.38646360622,3.674188037696849)--(6.38646360622,5.0);
\draw[red,thin] (6.046889156145355,4.18079410469)--(7.344788177427959,4.18079410469);
\draw[red,thin] (6.50207669924,5.0)--(6.50207669924,6.315502915750441);
\draw[red,thin] (6.74531989972,5.0)--(6.74531989972,6.315502915750441);
\draw[red,thin] (6.046889156145355,5.31033712283)--(7.344788177427959,5.31033712283);
\draw[red,thin] (6.046889156145355,5.67090577452)--(7.344788177427959,5.67090577452);
\draw[red,thin] (6.44109070603,6.315502915750441)--(6.44109070603,7.759944668133779);
\draw[red,thin] (7.08747444936,6.315502915750441)--(7.08747444936,7.759944668133779);
\draw[red,thin] (6.046889156145355,6.76544417257)--(7.344788177427959,6.76544417257);
\draw[red,thin] (6.046889156145355,7.20513424644)--(7.344788177427959,7.20513424644);
\draw[red,thin] (6.28571760474,7.759944668133779)--(6.28571760474,9.17952981019332);
\draw[red,thin] (7.08860844703,7.759944668133779)--(7.08860844703,9.17952981019332);
\draw[red,thin] (6.046889156145355,8.605429233)--(7.344788177427959,8.605429233);
\draw[red,thin] (7.91679980296,0.3476120732591395)--(7.91679980296,1.0998818656220464);
\draw[red,thin] (7.344788177427959,1.06737111646)--(8.66519005732634,1.06737111646);
\draw[red,thin] (7.344788177427959,1.48880599241)--(8.66519005732634,1.48880599241);
\draw[red,thin] (7.61099272975,1.0998818656220464)--(7.61099272975,2.2279279188574326);
\draw[red,thin] (8.39198585306,1.0998818656220464)--(8.39198585306,2.2279279188574326);
\draw[red,thin] (7.344788177427959,1.67545492003)--(8.66519005732634,1.67545492003);
\draw[red,thin] (7.344788177427959,1.87871857772)--(8.66519005732634,1.87871857772);
\draw[red,thin] (7.58140703076,2.2279279188574326)--(7.58140703076,3.674188037696849);
\draw[red,thin] (7.9921754603,2.2279279188574326)--(7.9921754603,3.674188037696849);
\draw[red,thin] (8.46616711874,2.2279279188574326)--(8.46616711874,3.674188037696849);
\draw[red,thin] (7.344788177427959,2.55213875053)--(8.66519005732634,2.55213875053);
\draw[red,thin] (7.344788177427959,3.11275350461)--(8.66519005732634,3.11275350461);
\draw[red,thin] (7.66638020828,3.674188037696849)--(7.66638020828,5.0);
\draw[red,thin] (7.88355052937,3.674188037696849)--(7.88355052937,5.0);
\draw[red,thin] (8.26812134644,3.674188037696849)--(8.26812134644,5.0);
\draw[red,thin] (8.40897862801,3.674188037696849)--(8.40897862801,5.0);
\draw[red,thin] (7.344788177427959,4.37089925066)--(8.66519005732634,4.37089925066);
\draw[red,thin] (7.344788177427959,4.77498357151)--(8.66519005732634,4.77498357151);
\draw[red,thin] (7.81034572031,5.0)--(7.81034572031,6.315502915750441);
\draw[red,thin] (8.08558854811,5.0)--(8.08558854811,6.315502915750441);
\draw[red,thin] (7.344788177427959,6.02797185414)--(8.66519005732634,6.02797185414);
\draw[red,thin] (7.80706347738,6.315502915750441)--(7.80706347738,7.759944668133779);
\draw[red,thin] (8.14568735608,6.315502915750441)--(8.14568735608,7.759944668133779);
\draw[red,thin] (7.344788177427959,6.66140997569)--(8.66519005732634,6.66140997569);
\draw[red,thin] (7.344788177427959,6.88848497234)--(8.66519005732634,6.88848497234);
\draw[red,thin] (7.344788177427959,7.2377233403)--(8.66519005732634,7.2377233403);
\draw[red,thin] (7.72554003362,7.759944668133779)--(7.72554003362,9.17952981019332);
\draw[red,thin] (7.344788177427959,8.31676362114)--(8.66519005732634,8.31676362114);
\draw[red,thin] (9.04393025441,0.3476120732591395)--(9.04393025441,1.0998818656220464);
\draw[red,thin] (8.66519005732634,0.865413070558)--(9.287239924156541,0.865413070558);
\draw[red,thin] (8.66519005732634,1.13517391829)--(9.287239924156541,1.13517391829);
\draw[red,thin] (8.79254600971,1.0998818656220464)--(8.79254600971,2.2279279188574326);
\draw[red,thin] (8.97732594248,1.0998818656220464)--(8.97732594248,2.2279279188574326);
\draw[red,thin] (9.07633594215,1.0998818656220464)--(9.07633594215,2.2279279188574326);
\draw[red,thin] (8.66519005732634,2.1095740559)--(9.287239924156541,2.1095740559);
\draw[red,thin] (8.99613943593,2.2279279188574326)--(8.99613943593,3.674188037696849);
\draw[red,thin] (8.66519005732634,2.42704685498)--(9.287239924156541,2.42704685498);
\draw[red,thin] (8.66519005732634,3.03780773059)--(9.287239924156541,3.03780773059);
\draw[red,thin] (8.66519005732634,3.40412430276)--(9.287239924156541,3.40412430276);
\draw[red,thin] (8.92538945339,3.674188037696849)--(8.92538945339,5.0);
\draw[red,thin] (8.66519005732634,4.1707049767)--(9.287239924156541,4.1707049767);
\draw[red,thin] (8.66519005732634,4.78813358594)--(9.287239924156541,4.78813358594);
\draw[red,thin] (9.12941938106,5.0)--(9.12941938106,6.315502915750441);
\draw[red,thin] (8.66519005732634,5.50208005411)--(9.287239924156541,5.50208005411);
\draw[red,thin] (9.10521171607,6.315502915750441)--(9.10521171607,7.759944668133779);
\draw[red,thin] (8.66519005732634,6.55715395953)--(9.287239924156541,6.55715395953);
\draw[red,thin] (8.66519005732634,7.12693221808)--(9.287239924156541,7.12693221808);
\draw[red,thin] (8.99502509486,7.759944668133779)--(8.99502509486,9.17952981019332);
\draw[red,thin] (8.66519005732634,8.33672037304)--(9.287239924156541,8.33672037304);
\fill[blue] (5.0,5.0) circle (1pt);\draw(5.0,5.0) circle (1pt);
\fill[blue] (0.0,0.0) circle (1pt);\draw(0.0,0.0) circle (1pt);
\fill[blue] (6.04688915615,6.99923657394) circle (1pt);\draw(6.04688915615,6.99923657394) circle (1pt);
\fill[blue] (6.04688915615,0.00797550549121) circle (1pt);\draw(6.04688915615,0.00797550549121) circle (1pt);
\fill[blue] (6.04688915615,7.82714116065) circle (1pt);\draw(6.04688915615,7.82714116065) circle (1pt);
\fill[blue] (6.04688915615,3.32243207143) circle (1pt);\draw(6.04688915615,3.32243207143) circle (1pt);
\fill[blue] (7.34478817743,1.98079273956) circle (1pt);\draw(7.34478817743,1.98079273956) circle (1pt);
\fill[blue] (7.34478817743,4.16944002465) circle (1pt);\draw(7.34478817743,4.16944002465) circle (1pt);
\fill[blue] (8.66519005733,4.98048181367) circle (1pt);\draw(8.66519005733,4.98048181367) circle (1pt);
\fill[blue] (8.66519005733,6.41931626302) circle (1pt);\draw(8.66519005733,6.41931626302) circle (1pt);
\fill[blue] (8.66519005733,8.21367214149) circle (1pt);\draw(8.66519005733,8.21367214149) circle (1pt);
\fill[blue] (8.66519005733,3.8359350464) circle (1pt);\draw(8.66519005733,3.8359350464) circle (1pt);
\fill[blue] (9.28723992416,7.86476986297) circle (1pt);\draw(9.28723992416,7.86476986297) circle (1pt);
\fill[blue] (9.28723992416,0.198847539956) circle (1pt);\draw(9.28723992416,0.198847539956) circle (1pt);
\fill[blue] (9.28723992416,8.33309452372) circle (1pt);\draw(9.28723992416,8.33309452372) circle (1pt);
\fill[blue] (9.28723992416,8.76866602085) circle (1pt);\draw(9.28723992416,8.76866602085) circle (1pt);
\fill[blue] (9.28723992416,5.58202022963) circle (1pt);\draw(9.28723992416,5.58202022963) circle (1pt);
\fill[blue] (9.28723992416,9.28708937098) circle (1pt);\draw(9.28723992416,9.28708937098) circle (1pt);
\fill[blue] (3.97448541664,0.987355999047) circle (1pt);\draw(3.97448541664,0.987355999047) circle (1pt);
\fill[blue] (3.97448541664,2.95747628708) circle (1pt);\draw(3.97448541664,2.95747628708) circle (1pt);
\fill[blue] (3.97448541664,1.34283771031) circle (1pt);\draw(3.97448541664,1.34283771031) circle (1pt);
\fill[blue] (3.01057701083,0.582962881331) circle (1pt);\draw(3.01057701083,0.582962881331) circle (1pt);
\fill[blue] (3.01057701083,4.32339253006) circle (1pt);\draw(3.01057701083,4.32339253006) circle (1pt);
\fill[blue] (2.0391312664,5.55703904763) circle (1pt);\draw(2.0391312664,5.55703904763) circle (1pt);
\fill[blue] (0.906418423524,7.62304851977) circle (1pt);\draw(0.906418423524,7.62304851977) circle (1pt);
\fill[blue] (0.906418423524,1.73560462234) circle (1pt);\draw(0.906418423524,1.73560462234) circle (1pt);
\fill[blue] (0.906418423524,4.36492882253) circle (1pt);\draw(0.906418423524,4.36492882253) circle (1pt);
\fill[blue] (4.18791091772,5.0) circle (1pt);\draw(4.18791091772,5.0) circle (1pt);
\fill[blue] (5.889028246,5.0) circle (1pt);\draw(5.889028246,5.0) circle (1pt);
\fill[blue] (4.18458770026,5.0) circle (1pt);\draw(4.18458770026,5.0) circle (1pt);
\fill[blue] (7.81469766927,5.0) circle (1pt);\draw(7.81469766927,5.0) circle (1pt);
\fill[blue] (9.99489256857,6.31550291575) circle (1pt);\draw(9.99489256857,6.31550291575) circle (1pt);
\fill[blue] (4.15607125026,6.31550291575) circle (1pt);\draw(4.15607125026,6.31550291575) circle (1pt);
\fill[blue] (8.16361993575,6.31550291575) circle (1pt);\draw(8.16361993575,6.31550291575) circle (1pt);
\fill[blue] (3.72129689898,7.75994466813) circle (1pt);\draw(3.72129689898,7.75994466813) circle (1pt);
\fill[blue] (9.87382119729,7.75994466813) circle (1pt);\draw(9.87382119729,7.75994466813) circle (1pt);
\fill[blue] (5.82293965792,7.75994466813) circle (1pt);\draw(5.82293965792,7.75994466813) circle (1pt);
\fill[blue] (0.599189920007,7.75994466813) circle (1pt);\draw(0.599189920007,7.75994466813) circle (1pt);
\fill[blue] (3.90234877199,7.75994466813) circle (1pt);\draw(3.90234877199,7.75994466813) circle (1pt);
\fill[blue] (3.51035167129,7.75994466813) circle (1pt);\draw(3.51035167129,7.75994466813) circle (1pt);
\fill[blue] (5.24559288765,7.75994466813) circle (1pt);\draw(5.24559288765,7.75994466813) circle (1pt);
\fill[blue] (3.340566537,9.17952981019) circle (1pt);\draw(3.340566537,9.17952981019) circle (1pt);
\fill[blue] (6.40025130918,9.17952981019) circle (1pt);\draw(6.40025130918,9.17952981019) circle (1pt);
\fill[blue] (4.10091076448,9.17952981019) circle (1pt);\draw(4.10091076448,9.17952981019) circle (1pt);
\fill[blue] (7.57734219699,3.6741880377) circle (1pt);\draw(7.57734219699,3.6741880377) circle (1pt);
\fill[blue] (4.26529514531,3.6741880377) circle (1pt);\draw(4.26529514531,3.6741880377) circle (1pt);
\fill[blue] (1.04130182706,3.6741880377) circle (1pt);\draw(1.04130182706,3.6741880377) circle (1pt);
\fill[blue] (3.87899560778,3.6741880377) circle (1pt);\draw(3.87899560778,3.6741880377) circle (1pt);
\fill[blue] (7.02649097955,2.22792791886) circle (1pt);\draw(7.02649097955,2.22792791886) circle (1pt);
\fill[blue] (2.65265895732,2.22792791886) circle (1pt);\draw(2.65265895732,2.22792791886) circle (1pt);
\fill[blue] (6.8214206936,2.22792791886) circle (1pt);\draw(6.8214206936,2.22792791886) circle (1pt);
\fill[blue] (9.80503345573,2.22792791886) circle (1pt);\draw(9.80503345573,2.22792791886) circle (1pt);
\fill[blue] (4.76850157066,1.09988186562) circle (1pt);\draw(4.76850157066,1.09988186562) circle (1pt);
\fill[blue] (9.68247254731,1.09988186562) circle (1pt);\draw(9.68247254731,1.09988186562) circle (1pt);
\fill[blue] (5.33186524021,0.347612073259) circle (1pt);\draw(5.33186524021,0.347612073259) circle (1pt);
\fill[blue] (0.0,0.0) circle (1pt);\draw(0.0,0.0) circle (1pt);
\fill[blue] (1.71772186696,3.27155380174) circle (1pt);\draw(1.71772186696,3.27155380174) circle (1pt);
\fill[blue] (1.41532545478,3.68747567048) circle (1pt);\draw(1.41532545478,3.68747567048) circle (1pt);
\fill[blue] (1.05504151445,4.22968009612) circle (1pt);\draw(1.05504151445,4.22968009612) circle (1pt);
\fill[blue] (1.48777314347,5.04181884322) circle (1pt);\draw(1.48777314347,5.04181884322) circle (1pt);
\fill[blue] (2.88338870641,0.774685778221) circle (1pt);\draw(2.88338870641,0.774685778221) circle (1pt);
\fill[blue] (2.44713843585,2.10054702587) circle (1pt);\draw(2.44713843585,2.10054702587) circle (1pt);
\fill[blue] (2.72124996687,3.9097767942) circle (1pt);\draw(2.72124996687,3.9097767942) circle (1pt);
\fill[blue] (3.25338520229,0.545296677725) circle (1pt);\draw(3.25338520229,0.545296677725) circle (1pt);
\fill[blue] (3.78596441254,1.03413825048) circle (1pt);\draw(3.78596441254,1.03413825048) circle (1pt);
\fill[blue] (3.54827366662,2.00967289189) circle (1pt);\draw(3.54827366662,2.00967289189) circle (1pt);
\fill[blue] (3.3011794688,3.71435695149) circle (1pt);\draw(3.3011794688,3.71435695149) circle (1pt);
\fill[blue] (4.97875172001,0.685133388264) circle (1pt);\draw(4.97875172001,0.685133388264) circle (1pt);
\fill[blue] (5.54085351273,2.99868958468) circle (1pt);\draw(5.54085351273,2.99868958468) circle (1pt);
\fill[blue] (4.82233174023,6.67367436837) circle (1pt);\draw(4.82233174023,6.67367436837) circle (1pt);
\fill[blue] (4.46458530774,8.46536412978) circle (1pt);\draw(4.46458530774,8.46536412978) circle (1pt);
\fill[blue] (7.16510147171,1.69250207506) circle (1pt);\draw(7.16510147171,1.69250207506) circle (1pt);
\fill[blue] (8.39198585306,1.78013885762) circle (1pt);\draw(8.39198585306,1.78013885762) circle (1pt);
\fill[blue] (7.81034572031,5.76416088671) circle (1pt);\draw(7.81034572031,5.76416088671) circle (1pt);
\fill[blue] (7.80706347738,7.17788363632) circle (1pt);\draw(7.80706347738,7.17788363632) circle (1pt);
\fill[blue] (7.57175095952,6.66140997569) circle (1pt);\draw(7.57175095952,6.66140997569) circle (1pt);
\fill[blue] (7.72554003362,8.0250694254) circle (1pt);\draw(7.72554003362,8.0250694254) circle (1pt);
\fill[blue] (7.72554003362,7.82745051607) circle (1pt);\draw(7.72554003362,7.82745051607) circle (1pt);
\fill[blue] (8.79254600971,1.30534292529) circle (1pt);\draw(8.79254600971,1.30534292529) circle (1pt);
\fill[blue] (8.99502509486,7.93122491016) circle (1pt);\draw(8.99502509486,7.93122491016) circle (1pt);
\draw[black,thin] (11.4541499086,4.1995267613)--(10.7230473275,4.08666159472);
\draw[black,thin] (11.4541499086,4.1995267613)--(9.32143223714,6.25713947464);
\draw[black,thin] (9.32143223714,6.25713947464)--(8.57137145932,5.69989903834);
\draw[black,thin] (8.57137145932,5.69989903834)--(9.82046990603,4.40820843004);
\draw[black,thin] (9.82046990603,4.40820843004)--(10.7230473275,4.08666159472);
\draw[black,thin] (4.59395545886,4.35853543336)--(4.59395545886,5.66796443819);
\draw[black,thin] (4.59395545886,4.35853543336)--(4.18624930899,4.33473867606);
\draw[black,thin] (4.18624930899,4.33473867606)--(4.18624930899,5.65809655566);
\draw[black,thin] (4.59395545886,5.66796443819)--(4.18624930899,5.65809655566);
\draw[black,thin] (10.893616693,9.02787769591)--(10.0733424899,8.55088027228);
\draw[black,thin] (10.893616693,9.02787769591)--(8.49736271207,9.02787769591);
\draw[black,thin] (10.0733424899,8.55088027228)--(8.92294057117,8.55088027228);
\draw[black,thin] (8.49736271207,9.02787769591)--(8.92294057117,8.55088027228);
\draw[black,thin] (10.0733424899,8.55088027228)--(9.63174302525,8.09893219335);
\draw[black,thin] (9.63174302525,8.09893219335)--(9.43461385352,6.99583593197);
\draw[black,thin] (9.32143223714,6.25713947464)--(9.37705147079,6.96954338605);
\draw[black,thin] (9.43461385352,6.99583593197)--(9.37705147079,6.96954338605);
\draw[black,thin] (8.23412545931,7.92625997073)--(7.94926031691,9.34550537372);
\draw[black,thin] (8.23412545931,7.92625997073)--(6.88619635697,7.92625997073);
\draw[black,thin] (7.94926031691,9.34550537372)--(7.83426381788,9.48780900416);
\draw[black,thin] (7.83426381788,9.48780900416)--(6.83705420566,8.34303983548);
\draw[black,thin] (6.88619635697,7.92625997073)--(6.83705420566,8.34303983548);
\draw[black,thin] (3.76031639549,0.527545189744)--(3.29832453023,1.03087239523);
\draw[black,thin] (3.76031639549,0.527545189744)--(3.9185274258,1.16509685468);
\draw[black,thin] (3.29832453023,1.03087239523)--(3.19055177671,1.40578911219);
\draw[black,thin] (3.19055177671,1.40578911219)--(3.42918847653,1.46393336742);
\draw[black,thin] (3.9185274258,1.16509685468)--(3.42918847653,1.46393336742);
\draw[black,thin] (5.92575806636,1.42586893019)--(5.97193153747,1.37072067752);
\draw[black,thin] (5.92575806636,1.42586893019)--(4.42512925016,0.665149000051);
\draw[black,thin] (6.15195352235,1.15163808685)--(6.01801015611,1.33421794961);
\draw[black,thin] (6.01801015611,1.33421794961)--(5.97193153747,1.37072067752);
\draw[black,thin] (4.3542779642,1.16509685468)--(3.9185274258,1.16509685468);
\draw[black,thin] (4.3542779642,1.16509685468)--(4.42512925016,0.665149000051);
\draw[black,thin] (6.92395583658,1.91572913365)--(6.92395583658,3.09501558141);
\draw[black,thin] (6.92395583658,1.91572913365)--(7.06816756872,1.95306251865);
\draw[black,thin] (7.06816756872,1.95306251865)--(7.71953849903,2.79199424211);
\draw[black,thin] (6.92395583658,3.09501558141)--(7.71953849903,2.79199424211);
\draw[black,thin] (6.90219370634,3.10639006515)--(6.16312644188,2.58338528414);
\draw[black,thin] (6.90219370634,3.10639006515)--(6.92395583658,3.09501558141);
\draw[black,thin] (6.01801015611,1.33421794961)--(6.92395583658,1.91572913365);
\draw[black,thin] (5.70654662901,1.82480951459)--(6.16312644188,2.58338528414);
\draw[black,thin] (5.70654662901,1.82480951459)--(5.92575806636,1.42586893019);
\draw[black,thin] (7.06816756872,1.95306251865)--(7.79544087827,1.49976533896);
\draw[black,thin] (7.71953849903,2.79199424211)--(8.03475082233,2.7487061257);
\draw[black,thin] (8.03475082233,2.7487061257)--(7.79544087827,1.49976533896);
\draw[black,thin] (2.95731115417,1.5313598006)--(3.19055177671,1.40578911219);
\draw[black,thin] (2.95731115417,1.5313598006)--(3.24585307275,2.71539844619);
\draw[black,thin] (3.42918847653,1.46393336742)--(4.50282865595,2.15015699869);
\draw[black,thin] (3.24585307275,2.71539844619)--(4.50282865595,2.15015699869);
\draw[black,thin] (5.25140809185,3.99420500421)--(5.25919155871,3.9963084837);
\draw[black,thin] (5.25140809185,3.99420500421)--(4.75556931797,3.05790190226);
\draw[black,thin] (5.25919155871,3.9963084837)--(6.16312644188,2.58338528414);
\draw[black,thin] (5.70654662901,1.82480951459)--(4.77807689664,2.20247039818);
\draw[black,thin] (4.75556931797,3.05790190226)--(4.77807689664,2.20247039818);
\draw[black,thin] (4.3542779642,1.16509685468)--(4.65569011828,2.15015699869);
\draw[black,thin] (4.50282865595,2.15015699869)--(4.65569011828,2.15015699869);
\draw[black,thin] (4.77807689664,2.20247039818)--(4.65569011828,2.15015699869);
\draw[black,thin] (8.92294057117,8.55088027228)--(8.9809689535,8.24862078939);
\draw[black,thin] (9.63174302525,8.09893219335)--(9.18682879432,8.09893219335);
\draw[black,thin] (8.9809689535,8.24862078939)--(9.18682879432,8.09893219335);
\draw[black,thin] (9.37705147079,6.96954338605)--(8.96983859349,7.14478714075);
\draw[black,thin] (9.18682879432,8.09893219335)--(8.96983859349,7.14478714075);
\draw[black,thin] (5.25058103683,9.31529650636)--(3.72073865074,8.53625496701);
\draw[black,thin] (5.25058103683,9.31529650636)--(5.25058103683,78.8360168149);
\draw[black,thin] (7.94926031691,9.34550537372)--(8.49736271207,9.02787769591);
\draw[black,thin] (7.83426381788,9.48780900416)--(5.25058103683,78.8360168149);
\draw[black,thin] (4.59395545886,4.35853543336)--(5.25140809185,3.99420500421);
\draw[black,thin] (4.59395545886,5.66796443819)--(4.84648951427,5.82997149047);
\draw[black,thin] (5.25919155871,3.9963084837)--(5.444514123,4.11195934924);
\draw[black,thin] (5.444514123,4.11195934924)--(5.444514123,5.89345456313);
\draw[black,thin] (4.84648951427,5.82997149047)--(5.444514123,5.89345456313);
\draw[black,thin] (8.9809689535,8.24862078939)--(8.39490717247,7.56423203247);
\draw[black,thin] (8.96983859349,7.14478714075)--(8.6112940096,7.22300652689);
\draw[black,thin] (8.6112940096,7.22300652689)--(8.39490717247,7.56423203247);
\draw[black,thin] (8.23412545931,7.92625997073)--(8.37674356707,7.57928015523);
\draw[black,thin] (6.88619635697,7.92625997073)--(6.88629090988,7.41318886729);
\draw[black,thin] (6.88629090988,7.41318886729)--(6.89603239021,7.39344451903);
\draw[black,thin] (6.89603239021,7.39344451903)--(8.37674356707,7.57928015523);
\draw[black,thin] (8.37674356707,7.57928015523)--(8.39490717247,7.56423203247);
\draw[black,thin] (7.709431478,6.2176738401)--(8.02920665349,6.76482950027);
\draw[black,thin] (7.709431478,6.2176738401)--(8.53046261214,5.69159544505);
\draw[black,thin] (8.53046261214,5.69159544505)--(8.55256415248,5.69989903834);
\draw[black,thin] (8.55256415248,5.69989903834)--(8.30826747589,6.88020886362);
\draw[black,thin] (8.30826747589,6.88020886362)--(8.02920665349,6.76482950027);
\draw[black,thin] (6.63136767696,5.93099717709)--(6.61397983334,5.94860807375);
\draw[black,thin] (6.63136767696,5.93099717709)--(7.709431478,6.2176738401);
\draw[black,thin] (6.61397983334,5.94860807375)--(6.90798545445,7.27567333129);
\draw[black,thin] (6.90798545445,7.27567333129)--(8.02920665349,6.76482950027);
\draw[black,thin] (8.24899334364,5.3845661795)--(6.85186295764,5.37660942557);
\draw[black,thin] (8.24899334364,5.3845661795)--(8.53046261214,5.69159544505);
\draw[black,thin] (6.85186295764,5.37660942557)--(6.63136767696,5.93099717709);
\draw[black,thin] (8.57137145932,5.69989903834)--(8.55256415248,5.69989903834);
\draw[black,thin] (8.24899334364,5.3845661795)--(8.22658663813,4.40820843004);
\draw[black,thin] (8.22658663813,4.40820843004)--(9.82046990603,4.40820843004);
\draw[black,thin] (8.6112940096,7.22300652689)--(8.30826747589,6.88020886362);
\draw[black,thin] (6.89603239021,7.39344451903)--(6.90798545445,7.27567333129);
\draw[black,thin] (6.78747515921,3.60551834195)--(6.39839409266,4.20172045972);
\draw[black,thin] (6.78747515921,3.60551834195)--(8.05500495391,4.20070858135);
\draw[black,thin] (6.39839409266,4.20172045972)--(6.85186295764,4.9965357932);
\draw[black,thin] (6.85186295764,4.9965357932)--(8.08044625372,4.30143493416);
\draw[black,thin] (8.05500495391,4.20070858135)--(8.08044625372,4.30143493416);
\draw[black,thin] (6.90219370634,3.10639006515)--(6.78747515921,3.60551834195);
\draw[black,thin] (5.444514123,4.11195934924)--(6.39839409266,4.20172045972);
\draw[black,thin] (8.25670246798,2.84416906403)--(8.03475082233,2.7487061257);
\draw[black,thin] (8.25670246798,2.84416906403)--(8.05500495391,4.20070858135);
\draw[black,thin] (6.85186295764,5.37660942557)--(6.85186295764,4.9965357932);
\draw[black,thin] (6.61397983334,5.94860807375)--(5.65042712056,6.02469076961);
\draw[black,thin] (5.444514123,5.89345456313)--(5.65042712056,6.02469076961);
\draw[black,thin] (8.22658663813,4.40820843004)--(8.08044625372,4.30143493416);
\draw[black,thin] (10.7230473275,4.08666159472)--(8.85757924175,2.76431578785);
\draw[black,thin] (8.85757924175,2.76431578785)--(8.25670246798,2.84416906403);
\draw[black,thin] (2.94351194474,1.52916891441)--(1.85382283822,1.1706267454);
\draw[black,thin] (2.94351194474,1.52916891441)--(2.18635765625,2.75078654135);
\draw[black,thin] (1.85382283822,1.1706267454)--(1.57041603654,2.36711837466);
\draw[black,thin] (2.18635765625,2.75078654135)--(1.57041603654,2.36711837466);
\draw[black,thin] (3.29832453023,1.03087239523)--(3.17303680722,0.828787734755);
\draw[black,thin] (2.95731115417,1.5313598006)--(2.94351194474,1.52916891441);
\draw[black,thin] (1.45287528797,0.345727330967)--(1.85382283822,1.1706267454);
\draw[black,thin] (9.35379167895,1.70627374217)--(9.10706161971,1.97704672166);
\draw[black,thin] (9.35379167895,1.70627374217)--(9.14427028702,0.798760429549);
\draw[black,thin] (9.10706161971,1.97704672166)--(7.83783465736,0.906267221957);
\draw[black,thin] (9.14427028702,0.798760429549)--(7.66506053053,0.137432714658);
\draw[black,thin] (7.83783465736,0.906267221957)--(7.65832479427,0.151687751534);
\draw[black,thin] (7.66506053053,0.137432714658)--(7.65832479427,0.151687751534);
\draw[black,thin] (8.85757924175,2.76431578785)--(9.10706161971,1.97704672166);
\draw[black,thin] (7.79544087827,1.49976533896)--(7.83783465736,0.906267221957);
\draw[black,thin] (6.15195352235,1.15163808685)--(7.65832479427,0.151687751534);
\draw[black,thin] (3.24585307275,2.71539844619)--(3.14523157663,2.89770871065);
\draw[black,thin] (2.54773718426,3.07453006858)--(2.18635765625,2.75078654135);
\draw[black,thin] (2.54773718426,3.07453006858)--(2.75977173308,3.06588264254);
\draw[black,thin] (2.75977173308,3.06588264254)--(3.14523157663,2.89770871065);
\draw[black,thin] (1.57041603654,2.36711837466)--(0.924427707511,2.70833575527);
\draw[black,thin] (1.50209255568,9.37308334166)--(2.17533130129,8.32021953107);
\draw[black,thin] (3.72073865074,8.53625496701)--(3.70878191014,8.51743636339);
\draw[black,thin] (2.17533130129,8.32021953107)--(3.61582428513,8.49250528376);
\draw[black,thin] (3.70878191014,8.51743636339)--(3.61582428513,8.49250528376);
\draw[black,thin] (3.74679956831,4.40278306849)--(3.62832487323,4.24430352978);
\draw[black,thin] (3.74679956831,4.40278306849)--(4.07214537654,4.32779271606);
\draw[black,thin] (3.62832487323,4.24430352978)--(3.56038569041,3.26702153262);
\draw[black,thin] (4.07214537654,4.32779271606)--(4.07214537654,3.33520491955);
\draw[black,thin] (4.07214537654,3.33520491955)--(3.56038569041,3.26702153262);
\draw[black,thin] (3.06653847701,3.97624610696)--(3.62832487323,4.24430352978);
\draw[black,thin] (3.06653847701,3.97624610696)--(2.75977173308,3.06588264254);
\draw[black,thin] (3.14523157663,2.89770871065)--(3.56038569041,3.26702153262);
\draw[black,thin] (4.18624930899,4.33473867606)--(4.07214537654,4.32779271606);
\draw[black,thin] (4.75556931797,3.05790190226)--(4.07214537654,3.33520491955);
\draw[black,thin] (2.17533130129,8.32021953107)--(2.24406028467,7.01290966833);
\draw[black,thin] (3.61582428513,8.49250528376)--(3.61582428513,6.94054358008);
\draw[black,thin] (3.61582428513,6.94054358008)--(2.8590041702,6.60221663462);
\draw[black,thin] (2.8590041702,6.60221663462)--(2.24406028467,7.01290966833);
\draw[black,thin] (4.18624930899,5.65809655566)--(3.20488828236,5.63682337389);
\draw[black,thin] (3.74679956831,4.40278306849)--(3.15230843886,5.43431016768);
\draw[black,thin] (3.20488828236,5.63682337389)--(3.15230843886,5.43431016768);
\draw[black,thin] (2.54773718426,3.07453006858)--(2.06136569443,3.83929013335);
\draw[black,thin] (0.924427707511,2.70833575527)--(1.24380137665,3.24487917542);
\draw[black,thin] (2.06136569443,3.83929013335)--(1.24380137665,3.24487917542);
\draw[black,thin] (4.18272414268,7.06468921706)--(4.57397082982,7.39604347449);
\draw[black,thin] (4.18272414268,7.06468921706)--(3.81182283548,6.99953878874);
\draw[black,thin] (3.81182283548,6.99953878874)--(3.81182283548,8.40886390459);
\draw[black,thin] (4.57397082982,7.39604347449)--(4.57397082982,7.8014133431);
\draw[black,thin] (3.81182283548,8.40886390459)--(4.57397082982,7.8014133431);
\draw[black,thin] (4.84648951427,5.82997149047)--(4.18272414268,7.06468921706);
\draw[black,thin] (5.65042712056,6.02469076961)--(5.36810336421,7.08661271146);
\draw[black,thin] (5.36810336421,7.08661271146)--(4.57397082982,7.39604347449);
\draw[black,thin] (3.61582428513,6.94054358008)--(3.81182283548,6.99953878874);
\draw[black,thin] (2.8590041702,6.60221663462)--(3.20488828236,5.63682337389);
\draw[black,thin] (3.70878191014,8.51743636339)--(3.81182283548,8.40886390459);
\draw[black,thin] (5.25058103683,9.31529650636)--(5.45091203167,8.77232180682);
\draw[black,thin] (5.45091203167,8.77232180682)--(4.57397082982,7.8014133431);
\draw[black,thin] (5.67926641655,8.64555523004)--(6.04904039562,7.41318886729);
\draw[black,thin] (5.67926641655,8.64555523004)--(5.53426627279,8.70452335386);
\draw[black,thin] (5.53426627279,8.70452335386)--(5.53426627279,7.26164137176);
\draw[black,thin] (6.04904039562,7.41318886729)--(5.53426627279,7.26164137176);
\draw[black,thin] (6.88629090988,7.41318886729)--(6.04904039562,7.41318886729);
\draw[black,thin] (6.83705420566,8.34303983548)--(5.67926641655,8.64555523004);
\draw[black,thin] (5.45091203167,8.77232180682)--(5.53426627279,8.70452335386);
\draw[black,thin] (5.36810336421,7.08661271146)--(5.53426627279,7.26164137176);
\draw[black,thin] (3.06653847701,3.97624610696)--(2.20142717051,4.58139739908);
\draw[black,thin] (2.06136569443,3.83929013335)--(1.97671883126,4.336554437);
\draw[black,thin] (1.97671883126,4.336554437)--(2.20142717051,4.58139739908);
\draw[black,thin] (3.15230843886,5.43431016768)--(2.27959456206,4.74708398931);
\draw[black,thin] (2.20142717051,4.58139739908)--(2.27959456206,4.74708398931);
\draw[black,thin] (2.24406028467,7.01290966833)--(0.867517807393,6.25820483091);
\draw[black,thin] (2.27959456206,4.74708398931)--(0.867517807393,6.25820483091);
\draw[black,thin] (1.50209255568,9.37308334166)--(0.0301133527034,6.06960129955);
\draw[black,thin] (0.0301133527034,6.06960129955)--(0.867517807393,6.25820483091);
\draw[black,thin] (1.24380137665,3.24487917542)--(1.21883237538,3.94771289993);
\draw[black,thin] (1.97671883126,4.336554437)--(1.81688160144,4.34510458743);
\draw[black,thin] (1.21883237538,3.94771289993)--(1.81688160144,4.34510458743);
\draw[black,thin] (0.674842007591,3.96116810147)--(1.28306583899,4.6295374695);
\draw[black,thin] (1.21883237538,3.94771289993)--(0.674842007591,3.96116810147);
\draw[black,thin] (1.81688160144,4.34510458743)--(1.28306583899,4.6295374695);
\fill[green] (5.0,5.0) circle (1.5pt);\draw(5.0,5.0) circle (1.5pt);

\draw[black,ultra thick] (5.444514123,5.89345456313) -- (4.84648951427,5.82997149047) -- (4.59395545886,5.66796443819) -- (4.59395545886,4.35853543336) -- (5.25140809185,3.99420500421) -- (5.25919155871,3.9963084837) -- (5.444514123,4.11195934924) -- cycle;
 \end{scope} 
 \draw (2.0,2.0) rectangle (8.0,8.0); 
 \end{tikzpicture}

%% file: statFig.tex
\begin{tikzpicture}[scale=1.5]
\draw[thick] (-3, -.2)--(3,-.2);

\fill[black] (1.3,-.2) circle (2pt);
\fill[black] (2.3,-.2) circle (2pt);
\fill[black] (-1.6,-.2) circle (2pt);
\fill[black] (-2.5,-.2) circle (2pt);

    \draw (0,-.15)--(0,-0.25);

    \draw [decorate,decoration={brace}] (-2.5,-.1)--(-1.6,-.1);
    \draw [decorate,decoration={brace}] (-1.6,-.1)--(0,-.1);
    \draw [decorate,decoration={brace}] (0,-.1)--(1.3,-.1);
    \draw [decorate,decoration={brace}] (1.3,-.1)--(2.3,-.1);

    \coordinate[label=below:{$o$}] (a) at (0,-.2);
    \coordinate[label=below:{$A_1$}] (a) at (1.3,-.2);
    \coordinate[label=below:{$A_2$}] (a) at (2.3,-.2);
    \coordinate[label=below:{$A_{-1}$}] (a) at (-1.5,-.2);
    \coordinate[label=below:{$A_{-2}$}] (a) at (-2.4,-.2);

    \coordinate[label=above:{$UI^*_{}$}] (a) at (.8,-.1);
    \coordinate[label=above:{$(1-U)I^*$}] (a) at (-.8,-.1);
    \coordinate[label=above:{$I_{1}$}] (a) at (1.8,-.1);
    \coordinate[label=above:{$I_{-1}$}] (a) at (-2,-.1);
\end{tikzpicture}

%% file: Convergence_speed_include.tex
	\begin{tikzpicture}[scale = 4.8]
		\draw[thick, ->] (0,0) -- (3.2,0);
		\draw[thick, ->] (0,0) -- (0,1.4);
		\coordinate[label = 180 : \Large{$\wh f(r)$}](a) at  (-0.02,1.4);
		\coordinate[label = 0 : \Large{$r$}](a) at  (3.25,0);
\foreach \x in {0,0.5,1,1.5,2,2.5,3}
 \node[anchor=north] at (\x,-0.05) {\x} ;

\foreach \x in {0,0.5,1,1.5,2,2.5,3}
\draw[thick] (\x,0) -- (\x,-0.05);

\foreach \y in {0,0.5,1}
 \node[anchor=north] at (-0.15,\y+0.06) {\y};

\foreach \y in {0,0.5,1}
 \draw[thick] (-0.05,\y) -- (0,\y);

\draw[thin] (1.5,0.5) -- (3.0,0.5) -- (3.0,1.35) -- (1.5,1.35) -- (1.5,0.5); 

\node[anchor=west] at (1.6,1.2) {\Large{$\lambda = .1$, $\mu = .2$, $\sigma = .02$}} ;

\draw[ultra thick, green](1.6,1.0) -- (1.9,1.0);
\node[anchor=west] at (2.0,1.0) {\Large{500 realizations}} ;

\draw[ultra thick, red](1.6,0.8) -- (1.9,0.8);
\node[anchor=west] at (2.0,0.8) {\Large{10,000 realizations}} ;

\draw[ultra thick, blue ](1.6,0.6) -- (1.9,0.6);
\node[anchor=west] at (2.0,0.6) {\Large{100,000 realizations}} ;

\draw [thick, green, xshift=0cm] 
		plot [smooth, tension=0.8] coordinates 
{(0.0,0.20135836352030803)(0.01,0.22169555823585915)(0.02,0.23820694404452442)(0.03,0.2581414220330349)(0.04,0.2742500911146596)(0.05,0.2899560434692436)(0.06,0.3098905214577541)(0.07,0.3259991905393787)(0.08,0.34915540234421416)(0.09,0.3682844468786434)(0.1,0.3912393003199585)(0.11,0.408556119582705)(0.12,0.4299001061158576)(0.13,0.4470155670150839)(0.14,0.4689636286387975)(0.15,0.48023969699593466)(0.16,0.4925225571706735)(0.17,0.5088325846158185)(0.18,0.5293711376948899)(0.19,0.5527287078632456)(0.2,0.5803148036655278)(0.21,0.6064913909231678)(0.22,0.6405209543580999)(0.23,0.6703219921591055)(0.24,0.7116004566807685)(0.25,0.75167077102131)(0.26,0.7943587440876152)(0.27,0.8448996933312125)(0.28,0.892218908758485)(0.29,0.9339000900071887)(0.3,0.9761853463464534)(0.31,1.0267262955900507)(0.32,1.0643803095683482)(0.33,1.1056587740900115)(0.34,1.1384801873438217)(0.35,1.1634486244203397)(0.36,1.1821749522277287)(0.37,1.1900279284050206)(0.38,1.192242870403744)(0.39,1.1853966860440535)(0.4,1.1670730749637055)(0.41,1.1423059962507076)(0.42,1.1193511428093925)(0.43,1.1036451904548084)(0.44,1.0967990060951178)(0.45,1.09800715627624)(0.46,1.115122617175466)(0.47,1.1288149858948469)(0.48,1.14391686315887)(0.49,1.162240474239218)(0.5,1.1757314845950788)(0.51,1.192846945494305)(0.52,1.2081501811218482)(0.53,1.2141909320274575)(0.54,1.223050700022351)(0.55,1.224862925294034)(0.56,1.2143922903909778)(0.57,1.1998944882175155)(0.58,1.1743219760504364)(0.59,1.148346747156317)(0.6,1.120357934626994)(0.61,1.094181347369354)(0.62,1.0782740366512498)(0.63,1.0565273333910563)(0.64,1.0392105141283097)(0.65,1.027531729044132)(0.66,1.0212896197750023)(0.67,1.0130339268706696)(0.68,1.0114230599625071)(0.69,1.0106176265084261)(0.7,1.0142420770517917)(0.71,1.0092081179637837)(0.72,1.0035700837852153)(0.73,0.9945089574268015)(0.74,0.9884682065211922)(0.75,0.9773934965275752)(0.76,0.9669228616245191)(0.77,0.9479951754536103)(0.78,0.9290674892827013)(0.79,0.9069180692954674)(0.8,0.8831577824000711)(0.81,0.8648341713197231)(0.82,0.8491282189651389)(0.83,0.8302005327942301)(0.84,0.8096619797151586)(0.85,0.793553310633534)(0.86,0.7744242660991048)(0.87,0.7577115219269192)(0.88,0.7446232282980991)(0.89,0.7287159175799949)(0.9,0.7120031734078093)(0.91,0.7051569890481187)(0.92,0.6912632619652175)(0.93,0.6795844768810396)(0.94,0.6668988999792602)(0.95,0.6503875141705949)(0.96,0.6338761283619297)(0.97,0.6123307834652567)(0.98,0.5940071723849087)(0.99,0.5835365374818527)(1.0,0.5672265100367078)(1.01,0.5497083324104409)(1.02,0.5315860796936133)(1.03,0.514470618794387)(1.04,0.5007782500750061)(1.05,0.4866831646285845)(1.06,0.4679568368211959)(1.07,0.45607669337349777)(1.08,0.440572099382434)(1.09,0.4258729388454516)(1.1,0.4123819284895909)(1.11,0.39929363486077085)(1.12,0.38982979177531635)(1.13,0.3779496483276182)(1.14,0.36465999633527785)(1.15,0.35439071979574216)(1.16,0.3411010678034018)(1.17,0.32922092435570366)(1.18,0.3169380641809649)(1.19,0.30606471255086826)(1.2,0.29337913564908885)(1.21,0.2849220843812359)(1.22,0.2726392242064971)(1.23,0.2611617974858395)(1.24,0.25149659603686475)(1.25,0.24283818640549149)(1.26,0.2325689098659558)(1.27,0.22491729205218408)(1.28,0.21565480733024991)(1.29,0.20921133969760003)(1.3,0.204580097336633)(1.31,0.19954613824862527)(1.32,0.1929013122524551)(1.33,0.18524969443868342)(1.34,0.1800143769871554)(1.35,0.17115460899226184)(1.36,0.16450978299609167)(1.37,0.15947582390808399)(1.38,0.15343507300247472)(1.39,0.1475956804603858)(1.4,0.14115221282773593)(1.41,0.13470874519508608)(1.42,0.1276612024718753)(1.43,0.12162045156626605)(1.44,0.1153783422971365)(1.45,0.10873351630096635)(1.46,0.10349819884943834)(1.47,0.09967238994255248)(1.48,0.0950411475815854)(1.49,0.09141669703821985)(1.5,0.08799360485837461)(1.51,0.0851745877690903)(1.52,0.08054334540812322)(1.53,0.0775229699553186)(1.54,0.0751066695930749)(1.55,0.0702740688685875)(1.56,0.06785776850634381)(1.57,0.06322652614537673)(1.58,0.05940071723849087)(1.59,0.05859528378440964)(1.6,0.05638034178568626)(1.61,0.054366758150483176)(1.62,0.05134638269767855)(1.63,0.05054094924359732)(1.64,0.04792329051783331)(1.65,0.046111065246150544)(1.66,0.04309068979334592)(1.67,0.041479822885183454)(1.68,0.03926488088646007)(1.69,0.03584178870661483)(1.7,0.03423092179845237)(1.71,0.03262005489028991)(1.72,0.030405112891566513)(1.73,0.02859288761988374)(1.74,0.026579303984680662)(1.75,0.024767078712997888)(1.76,0.023961645258916656)(1.77,0.023357570168355733)(1.78,0.02053855307907142)(1.79,0.019330402897949572)(1.8,0.018122252716827724)(1.81,0.017115460899226186)(1.82,0.015705952354584028)(1.83,0.015303235627543411)(1.84,0.014497802173462178)(1.85,0.013692368719380949)(1.86,0.013692368719380949)(1.87,0.013289651992340331)(1.88,0.012886935265299715)(1.89,0.012685576901779407)(1.9,0.011880143447698175)(1.91,0.011678785084177867)(1.92,0.010873351630096634)(1.93,0.010470634903056018)(1.94,0.01026927653953571)(1.95,0.009866559812495094)(1.96,0.00926248472193417)(1.97,0.008658409631373246)(1.98,0.008658409631373246)(1.99,0.007852976177292014)(2.0,0.006846184359690474)(2.01,0.006242109269129549)(2.02,0.006040750905609241)(2.03,0.005839392542088933)(2.04,0.005436675815048317)(2.05,0.005235317451528009)(2.06,0.005638034178568625)(2.07,0.005235317451528009)(2.08,0.005033959088007701)(2.09,0.004832600724487393)(2.1,0.004631242360967085)(2.11,0.004631242360967085)(2.12,0.004631242360967085)(2.13,0.004429883997446777)(2.14,0.003825808906885853)(2.15,0.003825808906885853)(2.16,0.003825808906885853)(2.17,0.0036244505433655444)(2.18,0.0036244505433655444)(2.19,0.0036244505433655444)(2.2,0.0032217338163249287)(2.21,0.0030203754528046207)(2.22,0.0030203754528046207)(2.23,0.0030203754528046207)(2.24,0.0024163003622436965)(2.25,0.0024163003622436965)(2.26,0.0020135836352030805)(2.27,0.0020135836352030805)(2.28,0.0018122252716827722)(2.29,0.0018122252716827722)(2.3,0.0018122252716827722)(2.31,0.0018122252716827722)(2.32,0.0018122252716827722)(2.33,0.0018122252716827722)(2.34,0.0018122252716827722)(2.35,0.0016108669081624644)(2.36,0.0016108669081624644)(2.37,0.0016108669081624644)(2.38,0.0016108669081624644)(2.39,0.0016108669081624644)(2.4,0.0016108669081624644)(2.41,0.0012081501811218483)(2.42,0.0012081501811218483)(2.43,0.0012081501811218483)(2.44,0.0012081501811218483)(2.45,0.0012081501811218483)(2.46,0.0010067918176015402)(2.47,0.0010067918176015402)(2.48,8.054334540812322E-4)(2.49,6.040750905609241E-4)(2.5,4.027167270406161E-4)(2.51,4.027167270406161E-4)(2.52,4.027167270406161E-4)(2.53,4.027167270406161E-4)(2.54,4.027167270406161E-4)(2.55,0.0)(2.56,0.0)(2.57,0.0)(2.58,0.0)(2.59,0.0)(2.6,0.0)(2.61,0.0)(2.62,0.0)(2.63,0.0)(2.64,0.0)(2.65,0.0)(2.66,0.0)(2.67,0.0)(2.68,0.0)(2.69,0.0)(2.7,0.0)(2.71,0.0)(2.72,0.0)(2.73,0.0)(2.74,0.0)(2.75,0.0)(2.76,0.0)(2.77,0.0)(2.78,0.0)(2.79,0.0)(2.8,0.0)(2.81,0.0)(2.82,0.0)(2.83,0.0)(2.84,0.0)(2.85,0.0)(2.86,0.0)(2.87,0.0)(2.88,0.0)(2.89,0.0)(2.9,0.0)(2.91,0.0)(2.92,0.0)(2.93,0.0)(2.94,0.0)(2.95,0.0)(2.96,0.0)(2.97,0.0)(2.98,0.0)(2.99,0.0)(3.0,0.0)};

\draw [thick, red, xshift=0cm] 
		plot [smooth, tension=0.8] coordinates 
{(0.0,0.19930281880952247)(0.01,0.21865512251592711)(0.02,0.23803732164515318)(0.03,0.2571404968280459)(0.04,0.27646290511162913)(0.05,0.2950279626837361)(0.06,0.3136926716652479)(0.07,0.33236734578770016)(0.08,0.3503743554671405)(0.09,0.36788310809955704)(0.1,0.3867271896179974)(0.11,0.4033291144248306)(0.12,0.42117668184922336)(0.13,0.4376490598238304)(0.14,0.4553870106978779)(0.15,0.4724274017060921)(0.16,0.48883998883505625)(0.17,0.506249090058068)(0.18,0.5239870409321156)(0.19,0.5449138369071155)(0.2,0.5715107980772461)(0.21,0.6040768786707221)(0.22,0.6419444142445314)(0.23,0.6841866466912098)(0.24,0.7263591231513047)(0.25,0.7711424665378044)(0.26,0.8154076225953993)(0.27,0.8580683909615776)(0.28,0.8992842138913869)(0.29,0.9386664508881484)(0.3,0.9784971192272316)(0.31,1.017500680868255)(0.32,1.054541109744005)(0.33,1.0919103882707903)(0.34,1.1271670569181949)(0.35,1.1570824100215042)(0.36,1.1818557503995277)(0.37,1.194900119890611)(0.38,1.1990954442265516)(0.39,1.1927277191655872)(0.4,1.1747007792042659)(0.41,1.1544715430950996)(0.42,1.1344017492409804)(0.43,1.120639889602183)(0.44,1.117062404004552)(0.45,1.121078355803564)(0.46,1.1330863506368376)(0.47,1.1468781056984567)(0.48,1.1622443530286708)(0.49,1.1784875327616469)(0.5,1.1915020068299087)(0.51,1.2069081147238847)(0.52,1.2200222402015515)(0.53,1.227595747316313)(0.54,1.2344916248471227)(0.55,1.2366241650083845)(0.56,1.2350795681626108)(0.57,1.227027734282706)(0.58,1.2113326373014561)(0.59,1.1901367825210634)(0.6,1.168781485485623)(0.61,1.1430714218591949)(0.62,1.1184475585952782)(0.63,1.0986567886874927)(0.64,1.07892580962535)(0.65,1.063838586241469)(0.66,1.0546009005896477)(0.67,1.046509206145981)(0.68,1.0407094941186241)(0.69,1.0360258778766003)(0.7,1.0298574556344455)(0.71,1.0226127981707194)(0.72,1.016703469593017)(0.73,1.0104752565052195)(0.74,0.9995434968935172)(0.75,0.9882031665032553)(0.76,0.9749993547571245)(0.77,0.9587561750241483)(0.78,0.9399818494922914)(0.79,0.918417284497101)(0.8,0.8961850550588988)(0.81,0.8725577058890299)(0.82,0.8493488926386609)(0.83,0.827664745952185)(0.84,0.8048246429166136)(0.85,0.7843662085658162)(0.86,0.7652829636648044)(0.87,0.7487009691398522)(0.88,0.7336735366016142)(0.89,0.7183272195532809)(0.9,0.7039674514580548)(0.91,0.6911921407723643)(0.92,0.677430281133567)(0.93,0.6650336458036147)(0.94,0.6515308798292695)(0.95,0.6360251205258887)(0.96,0.6197420802291507)(0.97,0.6028910268988055)(0.98,0.5864385792060794)(0.99,0.5694380487616271)(1.0,0.5528460890957344)(1.01,0.536742421335925)(1.02,0.5199112982874607)(1.03,0.503548536863199)(1.04,0.4877338581906634)(1.05,0.47156043444427065)(1.06,0.45619418711405646)(1.07,0.4428408982538185)(1.08,0.42998586644060427)(1.09,0.41644323990249726)(1.1,0.4036878594986878)(1.11,0.39119157275933075)(1.12,0.3792533339126403)(1.13,0.36778345669015233)(1.14,0.3559747646756881)(1.15,0.34368774589608103)(1.16,0.33168971620374776)(1.17,0.3201999086993788)(1.18,0.30882968288629553)(1.19,0.2979377838383551)(1.2,0.2874345252870933)(1.21,0.27645293997068865)(1.22,0.26567065747309343)(1.23,0.25525708519029594)(1.24,0.24565068932367692)(1.25,0.23650268994031984)(1.26,0.22714542259721276)(1.27,0.21847574997899855)(1.28,0.21032426468968904)(1.29,0.20283047870245102)(1.3,0.19533669271521298)(1.31,0.18766353419104637)(1.32,0.18106661088845116)(1.33,0.174240489344225)(1.34,0.1674044026590584)(1.35,0.16131570154442748)(1.36,0.15519710500697514)(1.37,0.14924791586551092)(1.38,0.14263106228103475)(1.39,0.13658222173016576)(1.4,0.13103163822632055)(1.41,0.12540133359495154)(1.42,0.11935249304408253)(1.43,0.11379194439929685)(1.44,0.10831111688203499)(1.45,0.10287014992853502)(1.46,0.09840576678720173)(1.47,0.09397127906868984)(1.48,0.09009483924284464)(1.49,0.08598920117536847)(1.5,0.08203304022199946)(1.51,0.07860503173847566)(1.52,0.07505744156366616)(1.53,0.07115110631499952)(1.54,0.06798219149592812)(1.55,0.06442463618017814)(1.56,0.06155467558932102)(1.57,0.05888401781727341)(1.58,0.056193429763344856)(1.59,0.05355266741411869)(1.6,0.05099162619241633)(1.61,0.04853023638011872)(1.62,0.046178463118166355)(1.63,0.04388648070185685)(1.64,0.041514777158023534)(1.65,0.03938223699676164)(1.66,0.03716001056703547)(1.67,0.03521680808364262)(1.68,0.03336329186871406)(1.69,0.0317589041772974)(1.7,0.030094725640237893)(1.71,0.02868964076763076)(1.72,0.02717493934467839)(1.73,0.025620377357964113)(1.74,0.0241654667806546)(1.75,0.022590974512059374)(1.76,0.021415087881083192)(1.77,0.020298992095749865)(1.78,0.019143035746654633)(1.79,0.018036905102261785)(1.8,0.017080251571976075)(1.81,0.01608373747792846)(1.82,0.015007502256357044)(1.83,0.014309942390523715)(1.84,0.013612382524690386)(1.85,0.01275538040380944)(1.86,0.012067785678916585)(1.87,0.011330365249321352)(1.88,0.010662700806309452)(1.89,0.009955175799535648)(1.9,0.009327371920285652)(1.91,0.008719498322916608)(1.92,0.008300962403416611)(1.93,0.007782775074511853)(1.94,0.007384169436892808)(1.95,0.0070453546449166194)(1.96,0.0065470975978928125)(1.97,0.006138526819333292)(1.98,0.005819642309238056)(1.99,0.0054708623763213915)(2.0,0.005151977866226156)(2.01,0.004862988778952349)(2.02,0.004504243705095208)(2.03,0.004235184899702353)(2.04,0.00403588208089283)(2.05,0.0037668232754999747)(2.06,0.0036173461613928327)(2.07,0.003308426792238073)(2.08,0.003069263409666646)(2.09,0.002830100027095219)(2.1,0.002670657772047601)(2.11,0.0026009017854642684)(2.12,0.0024414595304166503)(2.13,0.0022919824163095083)(2.14,0.0021225750203214145)(2.15,0.0019830630471547486)(2.16,0.0018236207921071307)(2.17,0.0016940739598809412)(2.18,0.0015844574095357038)(2.19,0.0015147014229523709)(2.2,0.0013951197316666574)(2.21,0.0012954683222618961)(2.22,0.0012157471947380873)(2.23,0.0011858517719166588)(2.24,0.001126060926273802)(2.25,0.0010363746578095168)(2.26,0.0010164443759285648)(2.27,9.66618671226184E-4)(2.28,9.167929665238034E-4)(2.29,8.470369799404706E-4)(2.3,8.171415571190423E-4)(2.31,7.872461342976138E-4)(2.32,7.27455288654757E-4)(2.33,6.776295839523764E-4)(2.34,6.278038792499959E-4)(2.35,5.879433154880913E-4)(2.36,5.779781745476151E-4)(2.37,5.181873289047584E-4)(2.38,4.882919060833301E-4)(2.39,4.583964832619017E-4)(2.4,4.384662013809495E-4)(2.41,3.98605637619045E-4)(2.42,3.786753557380927E-4)(2.43,3.687102147976166E-4)(2.44,3.5874507385714045E-4)(2.45,3.388147919761882E-4)(2.46,3.089193691547598E-4)(2.47,2.8898908727380756E-4)(2.48,2.6905880539285535E-4)(2.49,2.2919824163095085E-4)(2.5,2.1923310069047475E-4)(2.51,1.993028188095225E-4)(2.52,1.4947711410714186E-4)(2.53,1.4947711410714186E-4)(2.54,1.295468322261896E-4)(2.55,1.0961655034523737E-4)(2.56,1.0961655034523737E-4)(2.57,6.975598658333286E-5)(2.58,6.975598658333286E-5)(2.59,7.9721127523809E-5)(2.6,7.9721127523809E-5)(2.61,6.975598658333286E-5)(2.62,6.975598658333286E-5)(2.63,6.975598658333286E-5)(2.64,6.975598658333286E-5)(2.65,6.975598658333286E-5)(2.66,6.975598658333286E-5)(2.67,6.975598658333286E-5)(2.68,4.982570470238062E-5)(2.69,4.982570470238062E-5)(2.7,3.98605637619045E-5)(2.71,3.98605637619045E-5)(2.72,3.98605637619045E-5)(2.73,3.98605637619045E-5)(2.74,3.98605637619045E-5)(2.75,3.98605637619045E-5)(2.76,2.989542282142837E-5)(2.77,2.989542282142837E-5)(2.78,1.993028188095225E-5)(2.79,1.993028188095225E-5)(2.8,1.993028188095225E-5)(2.81,1.993028188095225E-5)(2.82,1.993028188095225E-5)(2.83,1.993028188095225E-5)(2.84,1.993028188095225E-5)(2.85,1.993028188095225E-5)(2.86,1.993028188095225E-5)(2.87,0.0)(2.88,0.0)(2.89,0.0)(2.9,0.0)(2.91,0.0)(2.92,0.0)(2.93,0.0)(2.94,0.0)(2.95,0.0)(2.96,0.0)(2.97,0.0)(2.98,0.0)(2.99,0.0)};
\draw [ blue, xshift=0cm] 
		plot [smooth, tension=0.8] coordinates 
{(0.0,0.19973522299898366)(0.01,0.2192893013305842)(0.02,0.23867060869429058)(0.03,0.25789013052736776)(0.04,0.2768460018660863)(0.05,0.2956640559009356)(0.06,0.3143373018991105)(0.07,0.33280182458925156)(0.08,0.35102067295510386)(0.09,0.369461227418485)(0.1,0.3877839381002968)(0.11,0.40512495016106853)(0.12,0.42275657697130387)(0.13,0.43997974525050626)(0.14,0.45692528156974)(0.15,0.4736700839898598)(0.16,0.4902271353003605)(0.17,0.5073374531785685)(0.18,0.5261814727924076)(0.19,0.5480135313423115)(0.2,0.5748769201595598)(0.21,0.6080978811248657)(0.22,0.6458748025267785)(0.23,0.6878102112715302)(0.24,0.7309590114960006)(0.25,0.7750315871268413)(0.26,0.818336180825251)(0.27,0.8610435662068937)(0.28,0.9029140610041706)(0.29,0.943332480730245)(0.3,0.9833454379536313)(0.31,1.0229988717756198)(0.32,1.0612331868382001)(0.33,1.0978077021976589)(0.34,1.1323399249019532)(0.35,1.162441021684015)(0.36,1.1860866760587498)(0.37,1.200183988098018)(0.38,1.203142066750633)(0.39,1.1948450655872551)(0.4,1.177196461283065)(0.41,1.1557299181912493)(0.42,1.1360360252035495)(0.43,1.1237563036935718)(0.44,1.1194819699213936)(0.45,1.1227027003922523)(0.46,1.1326794747810516)(0.47,1.1461975546736227)(0.48,1.1616241046219493)(0.49,1.1771734917324201)(0.5,1.1923783355832178)(0.51,1.2066054755174354)(0.52,1.2192347336676612)(0.53,1.2289688297605166)(0.54,1.2356250060669576)(0.55,1.2378480590989365)(0.56,1.2349189420536564)(0.57,1.2261775300191058)(0.58,1.2113901327843761)(0.59,1.1912068885003289)(0.6,1.1681145006933014)(0.61,1.1425004556959115)(0.62,1.1179450073804165)(0.63,1.0954658067079959)(0.64,1.0772749212733634)(0.65,1.062411624653894)(0.66,1.0513462932997504)(0.67,1.0430153371484627)(0.68,1.0366268060408403)(0.69,1.0312119841453378)(0.7,1.0258221291527103)(0.71,1.0203074396457084)(0.72,1.0143363551541538)(0.73,1.0071209202233153)(0.74,0.997956069516007)(0.75,0.9866780201493695)(0.76,0.9729771825277541)(0.77,0.9563911696099184)(0.78,0.9372175868781312)(0.79,0.9162144295036728)(0.8,0.8938690514306616)(0.81,0.8703512275986462)(0.82,0.846647650009242)(0.83,0.824025638652377)(0.84,0.8019199428469646)(0.85,0.7819074721785814)(0.86,0.7628057941270735)(0.87,0.7459860909983291)(0.88,0.7303997528716034)(0.89,0.7152758017861204)(0.9,0.7011435360828273)(0.91,0.6874876388863868)(0.92,0.673689929681617)(0.93,0.6601518762667459)(0.94,0.646184392122427)(0.95,0.6313031193328876)(0.96,0.61610826224324)(0.97,0.6007316461006632)(0.98,0.584092703348733)(0.99,0.5672969684467484)(1.0,0.5504323248928291)(1.01,0.5335087594481253)(1.02,0.5167729451130404)(1.03,0.5006103708679627)(1.04,0.48471843785004853)(1.05,0.46933183494632186)(1.06,0.4548190736432157)(1.07,0.4411192346977154)(1.08,0.42764609523031905)(1.09,0.41490498535521386)(1.1,0.40221680531420334)(1.11,0.39005692493802524)(1.12,0.378324477939065)(1.13,0.3664751858346502)(1.14,0.35479866469812965)(1.15,0.3429463765653699)(1.16,0.3309333015780961)(1.17,0.3193786189276049)(1.18,0.3078848555201284)(1.19,0.29674362478124505)(1.2,0.28566630931372144)(1.21,0.2748836033001213)(1.22,0.26458825123063867)(1.23,0.2546264569835644)(1.24,0.2450082073200483)(1.25,0.2359262467302845)(1.26,0.2269471497803652)(1.27,0.2186151949529626)(1.28,0.21041606404885432)(1.29,0.20278418117806318)(1.3,0.19523818445316157)(1.31,0.18802274952232326)(1.32,0.18114386844223826)(1.33,0.1743368920424329)(1.34,0.16752392358593757)(1.35,0.16093765460754608)(1.36,0.15449319763748387)(1.37,0.1481625897445311)(1.38,0.14189190241847802)(1.39,0.13580696884981394)(1.4,0.1298688406700542)(1.41,0.12408450861200362)(1.42,0.11845397267566227)(1.43,0.1132378873270438)(1.44,0.10809770136316497)(1.45,0.10308734329423545)(1.46,0.0984195311327492)(1.47,0.09390451641685718)(1.48,0.08978997082307812)(1.49,0.08585418825388315)(1.5,0.0820292587334526)(1.51,0.07836711341976624)(1.52,0.07481482247872932)(1.53,0.07141433030717162)(1.54,0.06819160248408301)(1.55,0.0649918442116393)(1.56,0.06206472451858919)(1.57,0.05917555451790889)(1.58,0.056432191230017845)(1.59,0.05381066642815619)(1.6,0.0512750277721841)(1.61,0.048802305711456676)(1.62,0.04648437844855347)(1.63,0.04421738366751501)(1.64,0.04196437035208647)(1.65,0.03988512668066705)(1.66,0.03779889227644266)(1.67,0.03582351092098272)(1.68,0.034040874055716784)(1.69,0.03230118026339564)(1.7,0.030673338195953925)(1.71,0.029036508043477252)(1.72,0.027554472688824792)(1.73,0.026001531330007693)(1.74,0.024632346376349663)(1.75,0.023305105819521416)(1.76,0.02201681363117797)(1.77,0.020930254018063497)(1.78,0.019832708967684083)(1.79,0.018773113609674475)(1.8,0.017802400425899416)(1.81,0.01684766605996427)(1.82,0.01592888403416895)(1.83,0.01512295240936805)(1.84,0.014376941351466846)(1.85,0.013597973981770808)(1.86,0.012870937770054507)(1.87,0.012257750635447628)(1.88,0.011585641610056048)(1.89,0.010935503459194355)(1.9,0.010289360012792643)(1.91,0.009705134485520617)(1.92,0.009194810990758213)(1.93,0.008677496763190845)(1.94,0.008149197098358535)(1.95,0.007690804761575866)(1.96,0.00726936344104801)(1.97,0.006840931387715191)(1.98,0.006468425196822085)(1.99,0.006099913710388961)(2.0,0.00574438501345077)(2.01,0.005390853668742569)(2.02,0.005030331591229404)(2.03,0.004750702279030827)(2.04,0.004482058404097194)(2.05,0.00421940658585353)(2.06,0.003994704459979673)(2.07,0.0037420394028859594)(2.08,0.0035183359531270973)(2.09,0.0032706642766083577)(2.1,0.0030779197864143383)(2.11,0.002901154114060238)(2.12,0.002729381822281112)(2.13,0.002558608206616981)(2.14,0.0023888332670678448)(2.15,0.0022350371453586275)(2.16,0.00208124102364941)(2.17,0.0019653945943099993)(2.18,0.0018425574321656246)(2.19,0.0017207189461362442)(2.2,0.0016048725167968339)(2.21,0.001505004905297342)(2.22,0.0013961492087628958)(2.23,0.0013042710061833634)(2.24,0.001223378240868775)(2.25,0.0011324987144042374)(2.26,0.0010565993296646238)(2.27,9.816986210400048E-4)(2.28,9.097939407603706E-4)(2.29,8.55865430550645E-4)(2.3,7.909514830759753E-4)(2.31,7.390203250962396E-4)(2.32,6.781010820815495E-4)(2.33,6.441460941717223E-4)(2.34,6.101911062618952E-4)(2.35,5.692453855471035E-4)(2.36,5.292983409473067E-4)(2.37,4.8935129634751E-4)(2.38,4.494042517477132E-4)(2.39,4.1345191160789617E-4)(2.4,3.9347838930799784E-4)(2.41,3.6152075362816047E-4)(2.42,3.445432596732468E-4)(2.43,3.18577680683379E-4)(2.44,2.9660680615349074E-4)(2.45,2.7463593162360255E-4)(2.46,2.546624093237042E-4)(2.47,2.2769815421884137E-4)(2.48,2.0972198414893285E-4)(2.49,1.8175905292907514E-4)(2.5,1.6977493954913611E-4)(2.51,1.5379612170921743E-4)(2.52,1.3482127552431398E-4)(2.53,1.248345143743648E-4)(2.54,1.1784378156940037E-4)(2.55,1.0486099207446642E-4)(2.56,1.028636398444766E-4)(2.57,9.387555480952232E-5)(2.58,8.988085034954265E-5)(2.59,8.388879365957314E-5)(2.6,7.58993847396138E-5)(2.61,6.491394747466969E-5)(2.62,6.091924301469002E-5)(2.63,5.7923214669705264E-5)(2.64,5.1931157979735753E-5)(2.65,4.294307294478149E-5)(2.66,4.194439682978657E-5)(2.67,4.0945720714791654E-5)(2.68,3.5952340139817066E-5)(2.69,3.5952340139817066E-5)(2.7,3.2956311794832304E-5)(2.71,3.195763567983739E-5)(2.72,2.8961607334852632E-5)(2.73,2.8961607334852632E-5)(2.74,2.5965578989867877E-5)(2.75,2.496690287487296E-5)(2.76,2.2969550644883125E-5)(2.77,2.2969550644883125E-5)(2.78,1.897484618490345E-5)(2.79,1.7976170069908533E-5)(2.8,1.5978817839918694E-5)(2.81,1.4980141724923774E-5)(2.82,1.3981465609928856E-5)(2.83,1.2982789494933938E-5)(2.84,1.2982789494933938E-5)(2.85,1.2982789494933938E-5)(2.86,1.198411337993902E-5)(2.87,9.986761149949185E-6)(2.88,9.986761149949185E-6)(2.89,7.989408919959347E-6)(2.9,6.990732804964428E-6)(2.91,5.99205668996951E-6)(2.92,4.993380574974592E-6)(2.93,4.993380574974592E-6)(2.94,4.993380574974592E-6)(2.95,3.9947044599796735E-6)(2.96,3.9947044599796735E-6)(2.97,3.9947044599796735E-6)(2.98,3.9947044599796735E-6)(2.99,3.9947044599796735E-6)};
		
	\end{tikzpicture}

%% file: Sigma_include.tex
	\begin{tikzpicture}[scale = 4.8]
		\draw[thick, ->] (0,0) -- (3.2,0);
		\draw[thick, ->] (0,0) -- (0,1.4);
		\coordinate[label = 180 : \Large{$\wh f(r)$}](a) at  (-0.02,1.4);
		\coordinate[label = 0 : \Large{$r$}](a) at  (3.25,0);
\foreach \x in {0,0.5,1,1.5,2,2.5,3}
 \node[anchor=north] at (\x,-0.05) {\x} ;

\foreach \x in {0,0.5,1,1.5,2,2.5,3}
\draw[thick] (\x,0) -- (\x,-0.05);

\foreach \y in {0,0.5,1}
 \node[anchor=north] at (-0.15,\y+0.06) {\y};

\foreach \y in {0,0.5,1}
 \draw[thick] (-0.05,\y) -- (0,\y);


\draw[thin] (1.9,0.4) -- (3.05,0.4) -- (3.05,1.42) -- (1.9,1.42) -- (1.9,0.4); 

\node[anchor=west] at (2.0,1.3) {\Large{$\lambda = .1$, $L\approx\mu = .2$}} ;

\draw[ultra thick, black](2.0,1.1) -- (2.4,1.1);
\node[anchor=west] at (2.5,1.1) {\Large{$\sigma = 0$}};

\draw[ultra thick, red](2.0,0.9) -- (2.4,0.9);
\node[anchor=west] at (2.5,0.9) {\Large{$\s=.002$}} ;


\draw[ultra thick, green](2.0,0.7) -- (2.4,0.7);
\node[anchor=west] at (2.5,0.7) {\Large{$\s=.02$}} ;


\draw[ultra thick, blue](2.0,0.5) -- (2.4,0.5);
\node[anchor=west] at (2.5,0.5)  {\Large{$\s=.2$}} ;


\draw [dashed] (0.2,0.555306284380677) -- (0.2,0);
\draw [dashed](0.4,1.172929256575152) -- (0.4,0);
\draw [dashed](0.6,1.1727998655926286) -- (0.6,0);
\draw [dashed](0.8,0.8973662764101924) -- (0.8,0);
\draw [dashed](1.0,0.5512951639224498) -- (1.0,0);
\draw[thick, <->] (0.4,0.27) -- (0.6,0.27);
\node[anchor=south] at (0.5,0.27) {$L$} ;

\draw [thick, green, xshift=0cm]  
		plot [smooth, tension=0.8] coordinates 
{(0.0,0.200141640238797)(0.01,0.22002571219652148)(0.02,0.23894910428109975)(0.03,0.2583628433842631)(0.04,0.2764856689078861)(0.05,0.29509884145009424)(0.06,0.31322166697371734)(0.07,0.3325253281747493)(0.08,0.3519690885239484)(0.09,0.3706823318862759)(0.1,0.38851495203155273)(0.11,0.40652769965304447)(0.12,0.4237098594675452)(0.13,0.4404917360015683)(0.14,0.45697340007523324)(0.15,0.4732449154266475)(0.16,0.4914778188524019)(0.17,0.5094505381458458)(0.18,0.5288242489209614)(0.19,0.5525710545352946)(0.2,0.5805908841687263)(0.21,0.6129938157233875)(0.22,0.6528520233769439)(0.23,0.6952019944514733)(0.24,0.7398636014707609)(0.25,0.7825337991696724)(0.26,0.8264248608740405)(0.27,0.8693952710333103)(0.28,0.9122355891264248)(0.29,0.9542553264945602)(0.3,0.9937132508676391)(0.31,1.0326908353041448)(0.32,1.0703674990790983)(0.33,1.1067132209464638)(0.34,1.1400468111282356)(0.35,1.1700380359180191)(0.36,1.1933045015957793)(0.37,1.2046425255153073)(0.38,1.2055231487323579)(0.39,1.1953259321621912)(0.4,1.1771930995565563)(0.41,1.1565184681198886)(0.42,1.1375550477072625)(0.43,1.1270976470047853)(0.44,1.1257867192612212)(0.45,1.1297094954099016)(0.46,1.1404370873267011)(0.47,1.1547272004397513)(0.48,1.1711388149393327)(0.49,1.1872902453066034)(0.5,1.2025710595388357)(0.51,1.2156102874003933)(0.52,1.228279253227509)(0.53,1.2387166397659624)(0.54,1.2447308960551382)(0.55,1.2458717034044995)(0.56,1.2416687289594848)(0.57,1.231291384913103)(0.58,1.2142893525748173)(0.59,1.194015004418627)(0.6,1.1713989990716431)(0.61,1.1453905929226114)(0.62,1.1194722505116872)(0.63,1.0960156502757001)(0.64,1.0761015570719399)(0.65,1.0619715572710808)(0.66,1.0512439653542813)(0.67,1.041647173704831)(0.68,1.035512832431512)(0.69,1.0297087248645866)(0.7,1.0255858070756676)(0.71,1.0203120748553751)(0.72,1.0133271316110413)(0.73,1.0058018059380625)(0.74,0.9970756304236509)(0.75,0.9859177339803379)(0.76,0.9706669409941416)(0.77,0.9537549723939632)(0.78,0.934331226208788)(0.79,0.9131362265074994)(0.8,0.8896996404355363)(0.81,0.8662530472815612)(0.82,0.8429165320297175)(0.83,0.8198101796641484)(0.84,0.7983249745845135)(0.85,0.7781707114124666)(0.86,0.7599678292327481)(0.87,0.7433160447648801)(0.88,0.726534168230857)(0.89,0.7113033894086845)(0.9,0.6964428726209538)(0.91,0.683263545611229)(0.92,0.6702843602417431)(0.93,0.6561943887689318)(0.94,0.6431751750713981)(0.95,0.6281745591355002)(0.96,0.6121031854243248)(0.97,0.5956315284326719)(0.98,0.5788796731446845)(0.99,0.5626181648752823)(1.0,0.5454860404708413)(1.01,0.529394652595642)(1.02,0.5124926910774755)(1.03,0.4960110270038106)(1.04,0.48105043939596054)(1.05,0.4657496109997045)(1.06,0.4519598519872514)(1.07,0.43814007172876246)(1.08,0.4247906243248347)(1.09,0.4117814177093129)(1.1,0.39938264309651944)(1.11,0.38688379766360653)(1.12,0.37571589413828166)(1.13,0.3635672965757867)(1.14,0.3518289893757813)(1.15,0.34020076007790717)(1.16,0.32859254494405693)(1.17,0.3171944785324574)(1.18,0.3056863342187266)(1.19,0.2940280836748167)(1.2,0.28319041385588584)(1.21,0.2732533814180296)(1.22,0.26255581074726586)(1.23,0.2525787499813618)(1.24,0.24319210705416225)(1.25,0.2335152587486164)(1.26,0.22503926028450336)(1.27,0.21665332555849776)(1.28,0.20914801404954286)(1.29,0.20118237676803874)(1.3,0.1942174476877286)(1.31,0.18664208660469014)(1.32,0.18000739123077403)(1.33,0.1732225896266788)(1.34,0.16660790841678658)(1.35,0.16027342550322865)(1.36,0.15386889301558715)(1.37,0.14734427554380236)(1.38,0.1412799838445668)(1.39,0.1348554371929014)(1.4,0.12883117382171363)(1.41,0.1229970450087527)(1.42,0.11729300826194698)(1.43,0.1120993326977502)(1.44,0.10705576336373251)(1.45,0.10177202406142828)(1.46,0.09694861053167327)(1.47,0.09244542362630033)(1.48,0.08788219422885576)(1.49,0.08367921978384103)(1.5,0.07984650737326807)(1.51,0.07623395076695778)(1.52,0.07265141540668331)(1.53,0.06914893670250437)(1.54,0.06630692541111345)(1.55,0.06299458126516136)(1.56,0.059732272529268965)(1.57,0.05691027540190193)(1.58,0.05418834909465429)(1.59,0.05143640154137083)(1.6,0.04888459562832617)(1.61,0.04661298801161582)(1.62,0.044381408722953236)(1.63,0.04213982235227871)(1.64,0.039808172243496726)(1.65,0.03779674875909681)(1.66,0.03572528278262527)(1.67,0.03395402926651191)(1.68,0.03235289614460154)(1.69,0.030651692202571763)(1.7,0.028960495342553925)(1.71,0.02740939763070325)(1.72,0.02602842031305555)(1.73,0.024547372175288452)(1.74,0.023316501087819852)(1.75,0.021965545016207972)(1.76,0.020844751830870707)(1.77,0.019683930317485686)(1.78,0.01865320087025588)(1.79,0.017582443094978317)(1.8,0.01665178446786791)(1.81,0.015651076266673926)(1.82,0.014770453049623218)(1.83,0.013929858160620272)(1.84,0.013019213697533746)(1.85,0.012168611726518858)(1.86,0.011558179723790528)(1.87,0.010997783131121895)(1.88,0.010527450276560723)(1.89,0.009876989945784632)(1.9,0.00937663584518764)(1.91,0.008906302990626466)(1.92,0.008385934726005595)(1.93,0.007865566461384723)(1.94,0.007405240688835489)(1.95,0.006874865342202677)(1.96,0.0063745112416056845)(1.97,0.005954213797104211)(1.98,0.005563937598638557)(1.99,0.005183668482184842)(2.0,0.004943498513898286)(2.01,0.0045832435614684515)(2.02,0.004232995691050557)(2.03,0.003882747820632662)(2.04,0.003732641590453564)(2.05,0.003542507032226707)(2.06,0.00333235830997597)(2.07,0.0030421529316297146)(2.08,0.002761954635295399)(2.09,0.0026118484051163007)(2.1,0.0025217846670088423)(2.11,0.002391692600853624)(2.12,0.002271607616710346)(2.13,0.002131508468543188)(2.14,0.0020214305664118496)(2.15,0.0019713951563521503)(2.16,0.0019013455822685715)(2.17,0.0017612464341014136)(2.18,0.0016211472859342557)(2.19,0.0015611047938626167)(2.2,0.001461033973743218)(2.21,0.0013109277435641204)(2.22,0.0011508144313730829)(2.23,0.001070757775277564)(2.24,9.907011191820452E-4)(2.25,9.006373810745865E-4)(2.26,8.305878069910076E-4)(2.27,7.905594789432482E-4)(2.28,7.705453149193684E-4)(2.29,7.004957408357895E-4)(2.3,6.504603307760903E-4)(2.31,6.00424920716391E-4)(2.32,5.403824286447519E-4)(2.33,4.603257725492331E-4)(2.34,4.0028328047759403E-4)(2.35,3.7026203444177446E-4)(2.36,3.3023370639401504E-4)(2.37,3.102195423701354E-4)(2.38,3.002124603581955E-4)(2.39,3.3023370639401504E-4)(2.4,2.801982963343158E-4)(2.41,2.5017705029849624E-4)(2.42,2.3016288627461656E-4)(2.43,2.1014872225073685E-4)(2.44,1.801274762149173E-4)(2.45,1.9013455822685716E-4)(2.46,1.801274762149173E-4)(2.47,1.7012039420297745E-4)(2.48,1.601133121910376E-4)(2.49,1.5010623017909776E-4)(2.5,1.3009206615521805E-4)(2.51,1.1007790213133835E-4)(2.52,1.1007790213133835E-4)(2.53,7.004957408357895E-5)(2.54,7.004957408357895E-5)(2.55,6.0042492071639104E-5)(2.56,6.0042492071639104E-5)(2.57,7.004957408357895E-5)(2.58,7.004957408357895E-5)(2.59,7.004957408357895E-5)(2.6,7.004957408357895E-5)(2.61,7.004957408357895E-5)(2.62,7.004957408357895E-5)(2.63,7.004957408357895E-5)(2.64,6.0042492071639104E-5)(2.65,5.0035410059699254E-5)(2.66,5.0035410059699254E-5)(2.67,4.00283280477594E-5)(2.68,4.00283280477594E-5)(2.69,4.00283280477594E-5)(2.7,3.0021246035819552E-5)(2.71,3.0021246035819552E-5)(2.72,3.0021246035819552E-5)(2.73,3.0021246035819552E-5)(2.74,3.0021246035819552E-5)(2.75,3.0021246035819552E-5)(2.76,2.00141640238797E-5)(2.77,2.00141640238797E-5)(2.78,2.00141640238797E-5)(2.79,1.000708201193985E-5)(2.8,1.000708201193985E-5)(2.81,1.000708201193985E-5)(2.82,0.0)(2.83,0.0)(2.84,0.0)(2.85,0.0)(2.86,0.0)(2.87,0.0)(2.88,0.0)(2.89,0.0)(2.9,0.0)(2.91,0.0)(2.92,0.0)(2.93,0.0)(2.94,0.0)(2.95,0.0)(2.96,0.0)(2.97,0.0)(2.98,0.0)(2.99,0.0)};

\draw [thick, blue, xshift=0cm]  
		plot [smooth, tension=0.8] coordinates 
{(0.0,0.19890083420998877)(0.01,0.2191489391325656)(0.02,0.23910863784553799)(0.03,0.25986393989535034)(0.04,0.28113638411410863)(0.05,0.3026574543756294)(0.06,0.3253818746841206)(0.07,0.3493593702481348)(0.08,0.372103680640047)(0.09,0.3953551881591947)(0.1,0.42030729781083775)(0.11,0.44573676946458485)(0.12,0.4708281097001749)(0.13,0.49702334956563043)(0.14,0.5250484771058178)(0.15,0.5526658579358747)(0.16,0.5798257668472487)(0.17,0.6064784786313873)(0.18,0.6335687722507878)(0.19,0.6621906022936052)(0.2,0.6896588074980046)(0.21,0.7176342098296394)(0.22,0.744664833198777)(0.23,0.7711285891904159)(0.24,0.7978807513916595)(0.25,0.823717969755537)(0.26,0.850241395997439)(0.27,0.8755515271506601)(0.28,0.9010307240129596)(0.29,0.9249386042850003)(0.3,0.9475635741763865)(0.31,0.9681199753919888)(0.32,0.9905360994074546)(0.33,1.0120273345438437)(0.34,1.0319571981316846)(0.35,1.0507831620896602)(0.36,1.0672720412456682)(0.37,1.0848946551566732)(0.38,1.0996530970550544)(0.39,1.1141629129106732)(0.4,1.126464929506561)(0.41,1.1375934311806097)(0.42,1.1476677584333457)(0.43,1.1548281884649052)(0.44,1.1632218036685669)(0.45,1.1706905299931518)(0.46,1.1766575550194516)(0.47,1.1810632084972028)(0.48,1.1829527664221977)(0.49,1.1843152371365362)(0.5,1.1847727090552191)(0.51,1.184713038804956)(0.52,1.1825450197120673)(0.53,1.180565956411678)(0.54,1.1779106302749744)(0.55,1.1731668453790662)(0.56,1.1684926757751315)(0.57,1.162744441666463)(0.58,1.1568171968070051)(0.59,1.1492191849401836)(0.6,1.141223371404942)(0.61,1.133038602077201)(0.62,1.1227156487817027)(0.63,1.1122037396937048)(0.64,1.1010354578528139)(0.65,1.088793111507189)(0.66,1.078360762752875)(0.67,1.0661581965740923)(0.68,1.053527993601758)(0.69,1.0400325720006103)(0.7,1.0260200082305166)(0.71,1.0121168399192384)(0.72,0.9986214183180906)(0.73,0.9842011078378664)(0.74,0.9689454138539603)(0.75,0.9525858202401887)(0.76,0.9355797989152347)(0.77,0.9194688313442256)(0.78,0.9036562150245315)(0.79,0.8862921721979995)(0.8,0.869693897583176)(0.81,0.8529265572592738)(0.82,0.8371636661481322)(0.83,0.8206350068252821)(0.84,0.8031914036650661)(0.85,0.7858373058802446)(0.86,0.7689705151392375)(0.87,0.752133559523362)(0.88,0.7351772634069604)(0.89,0.718817669793189)(0.9,0.7018315385516559)(0.91,0.6855813403966998)(0.92,0.6702858662459517)(0.93,0.6546224255519151)(0.94,0.638869479482484)(0.95,0.6227286767863434)(0.96,0.6077812790954626)(0.97,0.5907752577705087)(0.98,0.5741869281973956)(0.99,0.5598163429257239)(1.0,0.5439042761889248)(1.01,0.5286386371633082)(1.02,0.5147056337268985)(1.03,0.4995991153686498)(1.04,0.48421413584250717)(1.05,0.4697043199868885)(1.06,0.4551248888392963)(1.07,0.4412316655697286)(1.08,0.4280743753867378)(1.09,0.41498670049572056)(1.1,0.40150122393628335)(1.11,0.38887102096394904)(1.12,0.3760618072408258)(1.13,0.3646349543154619)(1.14,0.3524522782201001)(1.15,0.3401502616242123)(1.16,0.3294493967437149)(1.17,0.318619246320981)(1.18,0.30742112935495863)(1.19,0.2959147160959108)(1.2,0.2858901140517274)(1.21,0.2747218322108365)(1.22,0.265333712836125)(1.23,0.25570691246036154)(1.24,0.2458314860418356)(1.25,0.23667210262646562)(1.26,0.22770167500359514)(1.27,0.21875113746414565)(1.28,0.21034757721877362)(1.29,0.20227220334984808)(1.3,0.194743806775)(1.31,0.18730491557554643)(1.32,0.18014448554398682)(1.33,0.1735210877647942)(1.34,0.16613192177389313)(1.35,0.15969747978719998)(1.36,0.15308402704971785)(1.37,0.1471368921068392)(1.38,0.141478163373565)(1.39,0.13563047884779134)(1.4,0.13017065094872715)(1.41,0.12463126271597896)(1.42,0.11925099515059877)(1.43,0.11360221145903508)(1.44,0.10883853647970586)(1.45,0.10366711479024615)(1.46,0.0990725055199954)(1.47,0.09469668716737566)(1.48,0.0903208688147559)(1.49,0.08660142321502912)(1.5,0.08282230736503933)(1.51,0.07919236714070703)(1.52,0.07559226204150624)(1.53,0.07244962886098841)(1.54,0.06900864442915561)(1.55,0.0656472203310068)(1.56,0.06281288344351445)(1.57,0.05972992051325963)(1.58,0.05676629808353079)(1.59,0.05435959798958993)(1.6,0.05186339252025457)(1.61,0.04979482384447069)(1.62,0.04738812375052982)(1.63,0.04535933524158794)(1.64,0.043131645898436065)(1.65,0.040844286305021196)(1.66,0.03911384904739429)(1.67,0.03730385145608339)(1.68,0.035603249323587986)(1.69,0.03380319677398759)(1.7,0.03206281447465019)(1.71,0.03045171771754928)(1.72,0.028800840793606373)(1.73,0.027358809745583954)(1.74,0.02568804273822005)(1.75,0.024395187315855122)(1.76,0.023132167018621693)(1.77,0.022028267388756257)(1.78,0.02077519213323333)(1.79,0.019422666460605403)(1.8,0.018497777581528957)(1.81,0.01754305357732101)(1.82,0.01646898907258707)(1.83,0.01579272623627311)(1.84,0.014847947273775662)(1.85,0.014211464604303697)(1.86,0.013356191017200746)(1.87,0.01257053272207129)(1.88,0.011953940136020325)(1.89,0.011237897132864366)(1.9,0.010561634296550404)(1.91,0.009974876835630937)(1.92,0.00938811937471147)(1.93,0.008811306955502503)(1.94,0.008363780078530027)(1.95,0.007826747826163058)(1.96,0.007359330865769584)(1.97,0.007001309364191605)(1.98,0.0065836176123506285)(1.99,0.006235541152483148)(2.0,0.005937189901168165)(2.01,0.005648783691563681)(2.02,0.005420047732222194)(2.03,0.005042136147223215)(2.04,0.004863125396434225)(2.05,0.004674169603934736)(2.06,0.00442554356117225)(2.07,0.004206752643541262)(2.08,0.003928291475647278)(2.09,0.0036597753494637935)(2.1,0.0035006546820958023)(2.11,0.0032918088061753142)(2.12,0.003152578222228322)(2.13,0.0029835125131498315)(2.14,0.002784611678939843)(2.15,0.0026851612618348482)(2.16,0.0025260405944668575)(2.17,0.0023271397602568685)(2.18,0.0021879091763098765)(2.19,0.002058623634073384)(2.2,0.0018995029667053928)(2.21,0.0018099975913108979)(2.22,0.0016707670073639056)(2.23,0.0016210417988114084)(2.24,0.0015514265068379124)(2.25,0.0014420310480224186)(2.26,0.0014022508811804208)(2.27,0.0013127455057859259)(2.28,0.001272965338943928)(2.29,0.0012431302138124298)(2.3,0.0011436797967074354)(2.31,0.0010939545881549382)(2.32,9.74614087628945E-4)(2.33,9.447789624974466E-4)(2.34,9.049987956554489E-4)(2.35,8.751636705239506E-4)(2.36,8.254384619714533E-4)(2.37,7.25988044866459E-4)(2.38,6.563727528929629E-4)(2.39,6.165925860509651E-4)(2.4,5.569223357879686E-4)(2.41,5.469772940774691E-4)(2.42,4.873070438144725E-4)(2.43,4.176917518409764E-4)(2.44,3.7791158499897866E-4)(2.45,3.381314181569809E-4)(2.46,2.983512513149832E-4)(2.47,2.4862604276248594E-4)(2.48,2.088458759204882E-4)(2.49,1.6906570907849046E-4)(2.5,1.491756256574916E-4)(2.51,1.491756256574916E-4)(2.52,1.3923058394699214E-4)(2.53,1.0939545881549382E-4)(2.54,9.945041710499439E-5)(2.55,6.961529197349607E-5)(2.56,4.9725208552497194E-5)(2.57,4.9725208552497194E-5)(2.58,3.978016684199775E-5)(2.59,3.978016684199775E-5)(2.6,3.978016684199775E-5)(2.61,3.978016684199775E-5)(2.62,1.9890083420998876E-5)(2.63,1.9890083420998876E-5)(2.64,1.9890083420998876E-5)(2.65,1.9890083420998876E-5)(2.66,1.9890083420998876E-5)(2.67,1.9890083420998876E-5)(2.68,9.945041710499438E-6)(2.69,1.9890083420998876E-5)(2.7,1.9890083420998876E-5)(2.71,1.9890083420998876E-5)(2.72,1.9890083420998876E-5)(2.73,1.9890083420998876E-5)(2.74,1.9890083420998876E-5)(2.75,1.9890083420998876E-5)(2.76,1.9890083420998876E-5)(2.77,1.9890083420998876E-5)(2.78,1.9890083420998876E-5)(2.79,9.945041710499438E-6)(2.8,1.9890083420998876E-5)(2.81,1.9890083420998876E-5)(2.82,1.9890083420998876E-5)(2.83,1.9890083420998876E-5)(2.84,9.945041710499438E-6)(2.85,9.945041710499438E-6)(2.86,9.945041710499438E-6)(2.87,9.945041710499438E-6)(2.88,9.945041710499438E-6)(2.89,0.0)(2.9,0.0)(2.91,0.0)(2.92,0.0)(2.93,0.0)(2.94,0.0)(2.95,0.0)(2.96,0.0)(2.97,0.0)(2.98,0.0)(2.99,0.0)};

\draw [thin, red, xshift=0cm] 
		plot [smooth, tension=0.8] coordinates 
{(0.0,0.20157614402776752)(0.01,0.22218730475460674)(0.02,0.24213326420615436)(0.03,0.26266379447538246)(0.04,0.280765332209076)(0.05,0.3002879817581653)(0.06,0.3185709380214838)(0.07,0.3359064864078718)(0.08,0.3554291359569611)(0.09,0.3736012253410643)(0.1,0.3913096895939037)(0.11,0.40956240943561806)(0.12,0.42629322938992276)(0.13,0.44450563400283155)(0.14,0.4604200705738238)(0.15,0.4763748223736216)(0.16,0.49282343572628745)(0.17,0.5081734590940019)(0.18,0.5252973525291608)(0.19,0.541312577172167)(0.2,0.5597467155435063)(0.21,0.6044361466744623)(0.22,0.654537897272564)(0.23,0.7037728704513462)(0.24,0.7479986764510383)(0.25,0.7924059009803556)(0.26,0.8352408315862562)(0.27,0.8745179432500666)(0.28,0.9170706672543284)(0.29,0.9557934445220625)(0.3,0.995070556185873)(0.31,1.0346903472945308)(0.32,1.0704902704738621)(0.33,1.1095960424152491)(0.34,1.1455370688954)(0.35,1.1799058014521344)(0.36,1.215292493536209)(0.37,1.2476454646526658)(0.38,1.2843323228657193)(0.39,1.3172900224142594)(0.4,1.1791398121048289)(0.41,1.032150487879781)(0.42,1.0554526901293908)(0.43,1.0773136229492022)(0.44,1.0943770435411526)(0.45,1.1150889923400058)(0.46,1.1319810732095328)(0.47,1.1487925236214485)(0.48,1.1686780102297878)(0.49,1.1844211070783563)(0.5,1.1992470324715987)(0.51,1.214072957864841)(0.52,1.227387062177875)(0.53,1.2411748704293744)(0.54,1.2520499033996724)(0.55,1.2576738778180472)(0.56,1.261765873541811)(0.57,1.2641444720413384)(0.58,1.271532237719956)(0.59,1.276793375079081)(0.6,1.164969009179677)(0.61,1.0501814739630646)(0.62,1.0472888562962661)(0.63,1.045313410084794)(0.64,1.0415842514202804)(0.65,1.0385505304526625)(0.66,1.0360005922307112)(0.67,1.032674585854253)(0.68,1.0304169330411421)(0.69,1.0276351822535588)(0.7,1.0236641322162119)(0.71,1.0196326093356565)(0.72,1.0150064368302192)(0.73,1.0110555444072749)(0.74,1.005814564662553)(0.75,1.0012589438075254)(0.76,0.9954938660883313)(0.77,0.986090338969436)(0.78,0.9783094998099642)(0.79,0.9693998342439368)(0.8,0.8982736918237391)(0.81,0.8244968231095761)(0.82,0.8098825526675629)(0.83,0.7956714345136053)(0.84,0.7812788978300228)(0.85,0.7634797243123709)(0.86,0.747343553982948)(0.87,0.7320943186872475)(0.88,0.7196268341791301)(0.89,0.7066957245397488)(0.9,0.693673905635555)(0.91,0.6810149237906112)(0.92,0.670422097421952)(0.93,0.6587911539115499)(0.94,0.6473920229667796)(0.95,0.6358719463355926)(0.96,0.6244224213548154)(0.97,0.6135171519629132)(0.98,0.601936602488518)(0.99,0.5892373054147686)(1.0,0.5523791074792913)(1.01,0.514170349378828)(1.02,0.4998584431528565)(1.03,0.48734056460873215)(1.04,0.4753367052318786)(1.05,0.4637863921790875)(1.06,0.4520244241750673)(1.07,0.43939567875172764)(1.08,0.42808725707176987)(1.09,0.41703080557184685)(1.1,0.40408961712526414)(1.11,0.39155158096673703)(1.12,0.379225199759439)(1.13,0.3676144138634396)(1.14,0.35483448633207915)(1.15,0.3433547249296978)(1.16,0.3322478793937678)(1.17,0.3214433980738795)(1.18,0.3121305802197966)(1.19,0.30162846311594993)(1.2,0.2837385803334856)(1.21,0.26588901277982674)(1.22,0.25557839301280644)(1.23,0.24679975194039716)(1.24,0.23797071683198096)(1.25,0.2298673558420647)(1.26,0.22203612264658593)(1.27,0.21388236762066273)(1.28,0.20626278937641312)(1.29,0.19869360516817045)(1.3,0.1918097798496222)(1.31,0.18497634856708087)(1.32,0.17794134114051177)(1.33,0.17131956480919963)(1.34,0.1647481825138944)(1.35,0.15914436570992246)(1.36,0.1531071601962908)(1.37,0.14781578641556192)(1.38,0.14239338814121497)(1.39,0.13701130509567358)(1.4,0.12879707722654207)(1.41,0.12029056394857027)(1.42,0.11392075779729281)(1.43,0.10837741383652921)(1.44,0.10319690693501558)(1.45,0.0988529410312172)(1.46,0.09480126053625906)(1.47,0.09075965884850233)(1.48,0.08699018495518307)(1.49,0.08302921372503744)(1.5,0.07982415303499595)(1.51,0.0760446003344753)(1.52,0.07310158863166989)(1.53,0.07010818289285754)(1.54,0.06669146725158688)(1.55,0.06369806151277453)(1.56,0.0608659166891844)(1.57,0.05816479635921232)(1.58,0.05556446410125412)(1.59,0.05279279212087232)(1.6,0.04996064729728218)(1.61,0.047451024304136474)(1.62,0.045465499285462965)(1.63,0.04341950142358113)(1.64,0.04152468566972011)(1.65,0.03940813615742855)(1.66,0.03719079857312311)(1.67,0.03526574639765793)(1.68,0.033189512114171926)(1.69,0.03109312021628314)(1.7,0.02933940776324156)(1.71,0.02737404035897083)(1.72,0.026043637808387563)(1.73,0.024632604800193192)(1.74,0.02334251747841548)(1.75,0.021931484470221106)(1.76,0.020812736870867)(1.77,0.019825013765130937)(1.78,0.01867602974417266)(1.79,0.017617754988026882)(1.8,0.016589716653485267)(1.81,0.01582372730617975)(1.82,0.014997265115665903)(1.83,0.014301827418770105)(1.84,0.013374577156242375)(1.85,0.012588430194534083)(1.86,0.011812362040027178)(1.87,0.011076609114325826)(1.88,0.010421486646235581)(1.89,0.00986715225015922)(1.9,0.009322896661284249)(1.91,0.008849192722818994)(1.92,0.00819407025472875)(1.93,0.007619578244249613)(1.94,0.007105559076978805)(1.95,0.0065310670664996675)(1.96,0.0060069690920274725)(1.97,0.005543343960763607)(1.98,0.0052611373591247325)(1.99,0.004857985071069197)(2.0,0.0045153056262219925)(2.01,0.0041625473741734)(2.02,0.00385010435093036)(2.03,0.0036283705924998154)(2.04,0.0032856911476526105)(2.05,0.0030437997748192897)(2.06,0.0029530905100067944)(2.07,0.002811987209187357)(2.08,0.0026608051011665313)(2.09,0.0024289925355345985)(2.1,0.002277810427513773)(2.11,0.002086313090687394)(2.12,0.0019754462114721216)(2.13,0.0018444217178540728)(2.14,0.0017033184170346357)(2.15,0.0016226879594235285)(2.16,0.00155213630901381)(2.17,0.0014614270442013147)(2.18,0.0013707177793888191)(2.19,0.0012800085145763238)(2.2,0.0012396932857707704)(2.21,0.0011892992497638283)(2.22,0.0011389052137568865)(2.23,0.0010280383345416145)(2.24,9.675654913332842E-4)(2.25,8.970138409235655E-4)(2.26,8.163833833124584E-4)(2.27,7.559105401041282E-4)(2.28,7.055165040971863E-4)(2.29,5.946496248819142E-4)(2.3,5.240979744721956E-4)(2.31,4.7370393846525366E-4)(2.32,4.333887096597002E-4)(2.33,3.9307348085414667E-4)(2.34,3.7291586645136994E-4)(2.35,3.326006376458164E-4)(2.36,3.124430232430397E-4)(2.37,2.922854088402629E-4)(2.38,2.4189137283332104E-4)(2.39,2.4189137283332104E-4)(2.4,2.5197018003470943E-4)(2.41,2.5197018003470943E-4)(2.42,2.4189137283332104E-4)(2.43,2.3181256563193265E-4)(2.44,2.116549512291559E-4)(2.45,2.116549512291559E-4)(2.46,2.116549512291559E-4)(2.47,2.116549512291559E-4)(2.48,2.116549512291559E-4)(2.49,2.2173375843054428E-4)(2.5,1.9149733682637915E-4)(2.51,1.8141852962499076E-4)(2.52,1.5118210802082563E-4)(2.53,1.310244936180489E-4)(2.54,1.2094568641666052E-4)(2.55,1.1086687921527214E-4)(2.56,1.0078807201388376E-4)(2.57,9.070926481249538E-5)(2.58,9.070926481249538E-5)(2.59,9.070926481249538E-5)(2.6,9.070926481249538E-5)(2.61,1.0078807201388376E-4)(2.62,9.070926481249538E-5)(2.63,9.070926481249538E-5)(2.64,8.063045761110701E-5)(2.65,8.063045761110701E-5)(2.66,8.063045761110701E-5)(2.67,8.063045761110701E-5)(2.68,8.063045761110701E-5)(2.69,8.063045761110701E-5)(2.7,7.055165040971863E-5)(2.71,6.047284320833026E-5)(2.72,6.047284320833026E-5)(2.73,6.047284320833026E-5)(2.74,5.039403600694188E-5)(2.75,4.0315228805553506E-5)(2.76,4.0315228805553506E-5)(2.77,4.0315228805553506E-5)(2.78,2.0157614402776753E-5)(2.79,2.0157614402776753E-5)(2.8,2.0157614402776753E-5)(2.81,2.0157614402776753E-5)(2.82,2.0157614402776753E-5)(2.83,2.0157614402776753E-5)(2.84,2.0157614402776753E-5)(2.85,2.0157614402776753E-5)(2.86,2.0157614402776753E-5)(2.87,2.0157614402776753E-5)(2.88,2.0157614402776753E-5)(2.89,2.0157614402776753E-5)(2.9,2.0157614402776753E-5)(2.91,2.0157614402776753E-5)(2.92,2.0157614402776753E-5)(2.93,2.0157614402776753E-5)(2.94,2.0157614402776753E-5)(2.95,2.0157614402776753E-5)(2.96,1.0078807201388377E-5)(2.97,0.0)(2.98,0.0)(2.99,0.0)};

\draw [thin, black, xshift=0cm] 
		plot [smooth, tension=0.8] coordinates 
{(0.0,0.20152583268962623)(0.01,0.22007628558870632)(0.02,0.2405311576067034)(0.03,0.25984740867000405)(0.04,0.2793954144408978)(0.05,0.29783502813199864)(0.06,0.3172520421116441)(0.07,0.33663882721638616)(0.08,0.3549776779911421)(0.09,0.3722988233108155)(0.1,0.3898214944631785)(0.11,0.4068504273254519)(0.12,0.42372821581320813)(0.13,0.44118035292412977)(0.14,0.4588440921593755)(0.15,0.4757924146885731)(0.16,0.49203539680335695)(0.17,0.5090441770823614)(0.18,0.5246422765325385)(0.19,0.5408953349389568)(0.2,0.5543975657291618)(0.21,0.5968993638434039)(0.22,0.6468273888922589)(0.23,0.693349627368659)(0.24,0.7399423998865007)(0.25,0.7832402250398668)(0.26,0.828714529186281)(0.27,0.874017536374909)(0.28,0.9148869752443651)(0.29,0.9523808564162701)(0.3,0.9908420615850853)(0.31,1.0278018993003628)(0.32,1.0644796008498747)(0.33,1.102295923354083)(0.34,1.1421476567684568)(0.35,1.1790369604422928)(0.36,1.213084749875205)(0.37,1.248926119219055)(0.38,1.2819864320717884)(0.39,1.3151172789659629)(0.4,1.3426960891695383)(0.41,1.0252828263917424)(0.42,1.0477529567366357)(0.43,1.0694169837507705)(0.44,1.0903252888923194)(0.45,1.1094198615396613)(0.46,1.129532139642086)(0.47,1.1473772521267525)(0.48,1.1650510676536328)(0.49,1.1819590850162924)(0.5,1.1971037513429177)(0.51,1.2103439985506261)(0.52,1.2227277609694038)(0.53,1.2357060245946156)(0.54,1.2485532964285793)(0.55,1.2556067005727163)(0.56,1.2603727865158258)(0.57,1.2658845180398872)(0.58,1.271325715522507)(0.59,1.2796386561209543)(0.6,1.282480170361878)(0.61,1.0431984729178503)(0.62,1.0406693237175955)(0.63,1.0384726921412786)(0.64,1.0361148398988098)(0.65,1.032829968825969)(0.66,1.0294745637116867)(0.67,1.0266229731791285)(0.68,1.0222498626097636)(0.69,1.0190052967034606)(0.7,1.014803483091882)(0.71,1.0106923561050134)(0.72,1.0063797032854556)(0.73,1.0014221678012907)(0.74,0.9978350079794154)(0.75,0.993139456077747)(0.76,0.9869022315560031)(0.77,0.9787102064571698)(0.78,0.970236045192571)(0.79,0.9631624884651652)(0.8,0.9528745947063597)(0.81,0.8145472631482003)(0.82,0.8008233539420367)(0.83,0.7875931830259628)(0.84,0.7747962926501715)(0.85,0.7570217142069464)(0.86,0.7407182743423557)(0.87,0.7263595587632199)(0.88,0.7115272574772633)(0.89,0.6998689880561685)(0.9,0.6877875143864254)(0.91,0.6751216158018823)(0.92,0.6633726597560772)(0.93,0.651351643836141)(0.94,0.6406707747035908)(0.95,0.6299294478212337)(0.96,0.6190470528559938)(0.97,0.6089606849298781)(0.98,0.5960731079293765)(0.99,0.5853720862135573)(1.0,0.5747415985391795)(1.01,0.5085907439588098)(1.02,0.49503813171043237)(1.03,0.4829868869155927)(1.04,0.47136892266103575)(1.05,0.4588037869928376)(1.06,0.4475687218203909)(1.07,0.4357794606080478)(1.08,0.42483660789300104)(1.09,0.41323879622171306)(1.1,0.4002202274299632)(1.11,0.38844104250925454)(1.12,0.3769742226292148)(1.13,0.36430832404467184)(1.14,0.35092700875408067)(1.15,0.3398430879561512)(1.16,0.32926298173994584)(1.17,0.31996256456131955)(1.18,0.31055130817471405)(1.19,0.3019360788272325)(1.2,0.2921016181919788)(1.21,0.2659334888172308)(1.22,0.25683459747129417)(1.23,0.2480783000409299)(1.24,0.2393421551938346)(1.25,0.2309889094288496)(1.26,0.22335108036991275)(1.27,0.2158744719771276)(1.28,0.20850870279232178)(1.29,0.20104217069117114)(1.3,0.1949157853774065)(1.31,0.1886080268142212)(1.32,0.1815949278366222)(1.33,0.17444076077614046)(1.34,0.16810277333805174)(1.35,0.16193608285774916)(1.36,0.15577946866908107)(1.37,0.1504289578111715)(1.38,0.1454714223270067)(1.39,0.13976824126189027)(1.4,0.13459910365340136)(1.41,0.12361594577181673)(1.42,0.11716711912574869)(1.43,0.11164531131005294)(1.44,0.10678853874223294)(1.45,0.10277817467170938)(1.46,0.09837483522744105)(1.47,0.09421332678240027)(1.48,0.09012235237880085)(1.49,0.08602130168356696)(1.5,0.08263566769438124)(1.51,0.07931049145500241)(1.52,0.07574348421639603)(1.53,0.07302288547508606)(1.54,0.06951633598628656)(1.55,0.06613070199710085)(1.56,0.06361162908848052)(1.57,0.06076003855592231)(1.58,0.05801928723134339)(1.59,0.05540952769801273)(1.6,0.052104504041902866)(1.61,0.049343600134054984)(1.62,0.04680437464216569)(1.63,0.044587590482579806)(1.64,0.042270043406649103)(1.65,0.04036562428773213)(1.66,0.03836044225247035)(1.67,0.03625449730086376)(1.68,0.03439038334848472)(1.69,0.032889015894947)(1.7,0.030843528693147295)(1.71,0.02920109315672684)(1.72,0.02742766582905813)(1.73,0.02588599320898249)(1.74,0.02443500721361718)(1.75,0.022953792343348427)(1.76,0.021926010596631334)(1.77,0.020646321559052206)(1.78,0.019638692395604077)(1.79,0.018772131315038684)(1.8,0.017764502151590554)(1.81,0.016545270863818313)(1.82,0.015719014949790847)(1.83,0.014842377577590972)(1.84,0.01406650312173591)(1.85,0.01343169674876359)(1.86,0.012786814084156784)(1.87,0.01206132108647413)(1.88,0.01151720133821214)(1.89,0.010842089798701892)(1.9,0.010267741175536457)(1.91,0.009602705927660691)(1.92,0.009108967637571106)(1.93,0.008665610805653928)(1.94,0.008121491057391936)(1.95,0.007577371309129947)(1.96,0.007013098977598993)(1.97,0.006620123603854222)(1.98,0.006106232730495675)(1.99,0.005773715106557792)(2.0,0.005279976816468208)(2.01,0.004866848859454474)(2.02,0.004494026068978665)(2.03,0.004181661028309744)(2.04,0.0038894485709097862)(2.05,0.0035670072386063843)(2.06,0.0032747947812064265)(2.07,0.0030833452401512816)(2.08,0.0028515905325582113)(2.09,0.002639988408234104)(2.1,0.0023780048257375894)(2.11,0.0021865552846824445)(2.12,0.001964876868723856)(2.13,0.0018842665356480052)(2.14,0.001813732494206636)(2.15,0.0016525118280549351)(2.16,0.001551748911710122)(2.17,0.001450985995365309)(2.18,0.00122930757940672)(2.19,0.0010177054550826125)(2.2,9.471714136412433E-4)(2.21,8.161796223929862E-4)(2.22,7.657981642205797E-4)(2.23,7.053404144136918E-4)(2.24,6.44882664606804E-4)(2.25,5.642723315309535E-4)(2.26,4.937382900895843E-4)(2.27,4.53433123551659E-4)(2.28,4.131279570137338E-4)(2.29,4.030516653792525E-4)(2.3,3.9297537374477113E-4)(2.31,3.6274649884132724E-4)(2.32,3.1236504066892065E-4)(2.33,2.720598741309954E-4)(2.34,2.720598741309954E-4)(2.35,2.519072908620328E-4)(2.36,2.2167841595858886E-4)(2.37,2.3175470759307016E-4)(2.38,2.3175470759307016E-4)(2.39,2.3175470759307016E-4)(2.4,2.3175470759307016E-4)(2.41,1.9144954105514492E-4)(2.42,1.8137324942066362E-4)(2.43,1.712969577861823E-4)(2.44,1.712969577861823E-4)(2.45,1.712969577861823E-4)(2.46,1.5114437451721967E-4)(2.47,1.3099179124825705E-4)(2.48,1.3099179124825705E-4)(2.49,1.2091549961377574E-4)(2.5,9.068662471033181E-5)(2.51,7.053404144136919E-5)(2.52,6.045774980688787E-5)(2.53,6.045774980688787E-5)(2.54,4.030516653792525E-5)(2.55,3.0228874903443936E-5)(2.56,3.0228874903443936E-5)(2.57,2.0152583268962625E-5)(2.58,3.0228874903443936E-5)(2.59,2.0152583268962625E-5)(2.6,2.0152583268962625E-5)(2.61,2.0152583268962625E-5)(2.62,2.0152583268962625E-5)(2.63,2.0152583268962625E-5)(2.64,2.0152583268962625E-5)(2.65,2.0152583268962625E-5)(2.66,2.0152583268962625E-5)(2.67,2.0152583268962625E-5)(2.68,2.0152583268962625E-5)(2.69,2.0152583268962625E-5)(2.7,2.0152583268962625E-5)(2.71,2.0152583268962625E-5)(2.72,2.0152583268962625E-5)(2.73,0.0)(2.74,0.0)(2.75,0.0)(2.76,0.0)(2.77,0.0)(2.78,0.0)(2.79,0.0)(2.8,0.0)(2.81,0.0)(2.82,0.0)(2.83,0.0)(2.84,0.0)(2.85,0.0)(2.86,0.0)(2.87,0.0)(2.88,0.0)(2.89,0.0)(2.9,0.0)(2.91,0.0)(2.92,0.0)(2.93,0.0)(2.94,0.0)(2.95,0.0)(2.96,0.0)(2.97,0.0)(2.98,0.0)(2.99,0.0)};

\draw [thin, black, xshift=0cm] 
		plot [smooth, tension=0.8] coordinates 
{(0.0,0.19906305003609012)(0.01,0.21883996405717568)(0.02,0.2382983771982035)(0.03,0.2569605381390869)(0.04,0.27495583786234945)(0.05,0.29449387622339174)(0.06,0.31377313261938705)(0.07,0.3325547313902921)(0.08,0.35044054643603484)(0.09,0.3686747218193407)(0.1,0.38687903774514115)(0.11,0.4042871014707972)(0.12,0.42198380661900564)(0.13,0.4387349622795426)(0.14,0.45635204220773656)(0.15,0.4737800122383963)(0.16,0.4898344472238069)(0.17,0.505052817399066)(0.18,0.5209977677069568)(0.19,0.5371417810648837)(0.2,0.555306284380677)(0.21,0.5969801339057325)(0.22,0.6448846568969175)(0.23,0.6902013602376335)(0.24,0.7346123267006852)(0.25,0.7813125182391519)(0.26,0.8255642342621747)(0.27,0.8673177090072446)(0.28,0.9074090072845132)(0.29,0.947639649696807)(0.3,0.9869147894689275)(0.31,1.0262795076135645)(0.32,1.065703944673212)(0.33,1.1025803746923977)(0.34,1.1408104334518288)(0.35,1.1793689462438195)(0.36,1.214682731320222)(0.37,1.2475181814236749)(0.38,1.280194381087099)(0.39,1.2997622789056469)(0.4,1.172929256575152)(0.41,1.0399650923035455)(0.42,1.0485845223701082)(0.43,1.0694960957763995)(0.44,1.0899796836251132)(0.45,1.1106424282188594)(0.46,1.130618405289981)(0.47,1.148444501420713)(0.48,1.1661511597214231)(0.49,1.1839971621571586)(0.5,1.1987875467748401)(0.51,1.2126323819048501)(0.52,1.2267160926949034)(0.53,1.2402822395548632)(0.54,1.2513302388318661)(0.55,1.2613331570961797)(0.56,1.2682008323224248)(0.57,1.2729783455232908)(0.58,1.2785421577717997)(0.59,1.268389942219959)(0.6,1.1727998655926286)(0.61,1.0724621352219372)(0.62,1.0529639094709022)(0.63,1.0498983385003464)(0.64,1.0466038450222492)(0.65,1.0436875713392204)(0.66,1.0402238742685925)(0.67,1.036889568180488)(0.68,1.031743788337055)(0.69,1.029603860549167)(0.7,1.026000819343514)(0.71,1.022328106070348)(0.72,1.0187449711696983)(0.73,1.0139973174263377)(0.74,1.0084534114828325)(0.75,1.0037256640444754)(0.76,0.9973954590533277)(0.77,0.9892438271543499)(0.78,0.9791015647550111)(0.79,0.956189407695857)(0.8,0.8973662764101924)(0.81,0.8354178552389612)(0.82,0.8110823973720491)(0.83,0.7955057137067251)(0.84,0.780187812006448)(0.85,0.7630982491608497)(0.86,0.7488851473882728)(0.87,0.7334677141629776)(0.88,0.7199314267605235)(0.89,0.7062358889180406)(0.9,0.6932569780556874)(0.91,0.6816117896285762)(0.92,0.67021543001401)(0.93,0.6581323028768193)(0.94,0.6469549126172929)(0.95,0.6362453205253512)(0.96,0.62592390138098)(0.97,0.6152740282040491)(0.98,0.6029022596443061)(0.99,0.5833741744357657)(1.0,0.5512951639224498)(1.01,0.5202313749643179)(1.02,0.5020668716485247)(1.03,0.4884011932635471)(1.04,0.47505401575862727)(1.05,0.463000748078942)(1.06,0.4506588389767044)(1.07,0.43860557129701916)(1.08,0.42709972700493315)(1.09,0.4143497386501216)(1.1,0.40231637727543995)(1.11,0.39018348437574024)(1.12,0.3779809194085279)(1.13,0.3657982607463192)(1.14,0.3540435876416881)(1.15,0.3425377433496021)(1.16,0.33207698007020553)(1.17,0.3213872942832675)(1.18,0.3106378895813186)(1.19,0.29854480929162613)(1.2,0.2825699995262299)(1.21,0.2677796149085484)(1.22,0.2566918030215382)(1.23,0.2475747153298853)(1.24,0.23910458255084965)(1.25,0.23059463716180678)(1.26,0.22241314580532348)(1.27,0.21458001478640334)(1.28,0.20765262064514742)(1.29,0.20049630399634996)(1.3,0.19352909724508682)(1.31,0.1866614220188417)(1.32,0.18038098279020306)(1.33,0.1738517147490193)(1.34,0.16774047911291134)(1.35,0.16163919662930518)(1.36,0.15560758621321163)(1.37,0.14993428928718308)(1.38,0.14487808781626638)(1.39,0.13865736750263857)(1.4,0.13111287790627074)(1.41,0.12361815407241196)(1.42,0.1172979022337661)(1.43,0.11130610442767978)(1.44,0.10666793536183888)(1.45,0.10229850141354671)(1.46,0.09800869268526897)(1.47,0.09390799385452551)(1.48,0.08960823197374597)(1.49,0.08556725205801334)(1.5,0.0817253351923168)(1.51,0.07816210659667078)(1.52,0.07490742572858071)(1.53,0.07182194845302131)(1.54,0.06879619009247274)(1.55,0.06595954162945845)(1.56,0.06302336164142613)(1.57,0.06021657263591726)(1.58,0.05770837820546253)(1.59,0.05526985584252042)(1.6,0.05253273890452418)(1.61,0.049785668814026136)(1.62,0.04725756807856779)(1.63,0.044769279953116664)(1.64,0.04282841521526479)(1.65,0.04016097034478118)(1.66,0.038230058759431106)(1.67,0.036279240869077425)(1.68,0.03433837613122555)(1.69,0.0324671834608863)(1.7,0.03085477275559397)(1.71,0.0291627368302872)(1.72,0.02769962341252194)(1.73,0.02637585412978194)(1.74,0.024882881254511265)(1.75,0.023708409259298333)(1.76,0.0227529066191251)(1.77,0.021707825606435627)(1.78,0.020682650898749762)(1.79,0.01972714825857653)(1.8,0.018801505075908713)(1.81,0.01791567450324811)(1.82,0.01692035925306766)(1.83,0.01592504400288721)(1.84,0.015188510717753676)(1.85,0.01422305492507864)(1.86,0.013327271199916233)(1.87,0.012690269439800745)(1.88,0.012182658662208716)(1.89,0.011535703749591422)(1.9,0.010948467751984957)(1.91,0.010281606534364055)(1.92,0.009644604774248566)(1.93,0.009156900301660146)(1.94,0.008639336371566312)(1.95,0.008111819288970673)(1.96,0.007564395901371424)(1.97,0.007176222953801049)(1.98,0.006688518481212628)(1.99,0.006379970753656688)(2.0,0.006051516721097139)(2.01,0.00554390594350511)(2.02,0.005155732995934734)(2.03,0.004737700590858945)(2.04,0.004508778083317441)(2.05,0.0041305582882488695)(2.06,0.0038220105606929303)(2.07,0.0035731817481478175)(2.08,0.0033741186981117274)(2.09,0.003105383580563006)(2.1,0.002946133140534134)(2.11,0.002796835853007066)(2.12,0.0026276322604763896)(2.13,0.0024285692104402995)(2.14,0.0023290376854222545)(2.15,0.0021797403978951867)(2.16,0.002110068330382555)(2.17,0.002030443110368119)(2.18,0.0019806773478590967)(2.19,0.0018811458228410517)(2.2,0.0018015206028266155)(2.21,0.0017019890778085705)(2.22,0.0016124107052923298)(2.23,0.001552691790281503)(2.24,0.0014432071127616533)(2.25,0.0013635818927472173)(2.26,0.0012740035202309767)(2.27,0.0012242377577219542)(2.28,0.0011645188427111272)(2.29,0.001084893622696691)(2.3,9.555026401732326E-4)(2.31,9.057368776642101E-4)(2.32,8.559711151551875E-4)(2.33,8.360648101515785E-4)(2.34,7.86299047642556E-4)(2.35,7.26580132631729E-4)(2.36,6.967206751263154E-4)(2.37,6.668612176209019E-4)(2.38,6.370017601154883E-4)(2.39,6.170954551118794E-4)(2.4,5.971891501082704E-4)(2.41,5.474233875992478E-4)(2.42,5.374702350974433E-4)(2.43,5.175639300938343E-4)(2.44,5.076107775920297E-4)(2.45,4.777513200866163E-4)(2.46,4.4789186258120277E-4)(2.47,4.3793871007939824E-4)(2.48,4.2798555757759377E-4)(2.49,4.080792525739847E-4)(2.5,3.8817294757037573E-4)(2.51,3.583134900649622E-4)(2.52,3.483603375631577E-4)(2.53,3.384071850613532E-4)(2.54,3.1850088005774417E-4)(2.55,2.985945750541352E-4)(2.56,2.6873511754872166E-4)(2.57,2.3887566004330814E-4)(2.58,2.1896935503969912E-4)(2.59,2.1896935503969912E-4)(2.6,2.1896935503969912E-4)(2.61,1.9906305003609013E-4)(2.62,1.791567450324811E-4)(2.63,1.2939098252345857E-4)(2.64,1.2939098252345857E-4)(2.65,1.2939098252345857E-4)(2.66,1.2939098252345857E-4)(2.67,1.1943783002165407E-4)(2.68,1.1943783002165407E-4)(2.69,1.0948467751984956E-4)(2.7,9.953152501804506E-5)(2.71,8.957837251624055E-5)(2.72,8.957837251624055E-5)(2.73,8.957837251624055E-5)(2.74,7.962522001443604E-5)(2.75,6.967206751263155E-5)(2.76,5.9718915010827035E-5)(2.77,4.976576250902253E-5)(2.78,4.976576250902253E-5)(2.79,4.976576250902253E-5)(2.8,4.976576250902253E-5)(2.81,4.976576250902253E-5)(2.82,4.976576250902253E-5)(2.83,4.976576250902253E-5)(2.84,2.9859457505413518E-5)(2.85,9.953152501804505E-6)(2.86,9.953152501804505E-6)(2.87,9.953152501804505E-6)(2.88,9.953152501804505E-6)(2.89,9.953152501804505E-6)(2.9,9.953152501804505E-6)(2.91,9.953152501804505E-6)(2.92,9.953152501804505E-6)(2.93,9.953152501804505E-6)(2.94,0.0)(2.95,0.0)(2.96,0.0)(2.97,0.0)(2.98,0.0)(2.99,0.0)};
		
	\end{tikzpicture}

%% file: zSmall.tex
\begin{tikzpicture}[scale=1.5]

	\draw[thick] (-1,0)--(1,0);
	\fill (.45, -.05) rectangle (.55, .05);

	\draw[thick] (1,0)--(2.0,0);
	\draw[thick] (1,0)--(1,-1.0);
	\draw[thick] (1,0)--(1,1.0);

	\fill[blue] (-.2,0) circle (1.3pt);
	\fill[blue] (1.2,0) circle (1.3pt);
	\fill[blue] (1,.2) circle (1.3pt);
	\fill[blue] (1,-.2) circle (1.3pt);

	\draw[ultra thick, dashed, blue] (-.2, 0)--(1.2, 0);
	\draw[ultra thick, dashed, blue] (1, -.2)--(1, 0.2);

	\coordinate[label=above :$o$] (A) at (.5, 0);
	\draw[very thick, decorate, decoration={brace}] (.5, -.05) -- (-.2, -.05);
	\coordinate[label=below :$x$] (B) at (.25, -.05);

\end{tikzpicture}

%% file: zLarge.tex
\begin{tikzpicture}[scale=1.5]

	\draw[thick] (-1,0)--(1,0);
	\fill (.45, -.05) rectangle (.55, .05);

	\draw[thick] (-1, -2) rectangle (1, 2);

	\draw[ultra thick, dashed, blue] (-1, 0)--(1, 0);
	\draw[ultra thick, dashed, blue] (-1, -2)--(-1, 2);
	\draw[ultra thick, dashed, blue] (1, -2)--(1, 2);
	\draw[ultra thick, dashed, blue] (-.4, -2)--(1, -2);
	\draw[ultra thick, dashed, blue] (-.4, 2)--(1, 2);
	\draw[ultra thick, dashed, blue] (-.6, -2)--(-1, -2);
	\draw[ultra thick, dashed, blue] (-.6, 2)--(-1, 2);

	\fill[blue] (-.55, -2) circle (1.3pt);
	\fill[blue] (-.55, 2) circle (1.3pt);
	\fill[blue] (-.45, -2) circle (1.3pt);
	\fill[blue] (-.45, 2) circle (1.3pt);

	\coordinate[label=above :$o$] (A) at (.5, 0);

\end{tikzpicture}

%% file: longDist.tex
\begin{tikzpicture}[scale=1.5]

	\draw[thick] (-4.2,0)--(1.2,0);
	\fill (.45, -.05) rectangle (.55, .05);

	\draw[thick] (-4.2, .7) -- (1.2, .7);
	\draw[thick] (-4.2, .5) -- (1.2, .5);
	\draw[thick] (-4, .9) -- (-4, -.2);

	\draw[ultra thick, dashed, blue] (.5, 0)--(-4, 0);
	\draw[ultra thick, dashed, blue] (-4, .7)--(-4, 0);
	\draw[ultra thick, dashed, blue] (.55, .7)--(-4, .7);

	\fill[blue] (.55, .7) circle (1.3pt);

	\coordinate[label=above :$o$] (A) at (.5, 0);

\end{tikzpicture}